\documentclass[11pt]{article}

\usepackage[paperwidth=8.5in, 
			paperheight=11in, 
			portrait, 
			top=1in, 
			bottom=1.in, 
			left=1.15in, 
			right=1.15in
		    ]{geometry}
\usepackage[utf8]{inputenc}
\usepackage{authblk}

\usepackage{amsfonts,amsmath,amssymb,amsthm}
\usepackage{latexsym,mathrsfs}
\usepackage{braket}

\usepackage{graphicx,subcaption,epsfig,caption,float}
\usepackage{xcolor}
\usepackage{enumitem}
\usepackage{chngcntr}

\usepackage{tikz}
\usetikzlibrary{arrows,calc,decorations.markings}

\usepackage[hidelinks]{hyperref}
\usepackage{bookmark}

\newcommand{\Tr}{\text{Tr}}
\newcommand{\id}{\text{id}}

\newcommand{\bI}{\mathbb{I}}

\newcommand{\bR}{\mathbb{R}}

\newcommand{\cB}{\mathcal{B}}
\newcommand{\cR}{\mathcal{R}}

\newcommand{\Q}{Q} 
\newcommand{\A}{\mathbb{A}} 
\newcommand{\C}{C} 

\newcommand{\E}{I_{\text{CS}}} 
\newcommand{\Y}{I_{\text{RT}}} 

\newcommand{\gsl}{\mathfrak{sl}}
\newcommand{\su}{\mathfrak{su}}

\newcommand{\saw}{\mathfrak{\textbf{saw}}}

\newcommand{\qi}{{q^{-1}}}
\newcommand{\Usl}{U_q(\su_2)}
\newcommand{\Ut}{U_q(\su_2)^{\otimes 3}}

\newtheorem{prop}{Proposition}[section]
\newtheorem{thm}{Theorem}[section]
\newtheorem{rem}{Remark}[section]
\newtheorem{coro}{Corollary}[section]
\newtheorem{defi}{Definition}[section]

\numberwithin{equation}{section}
\counterwithin{figure}{section}

\def\colfund{blue} 
\def\lw{0.7pt} 
\def\eseq{0.1cm} 

\def\ra{0.25cm} 
\def\es{0.07cm} 
\def\lx{0.15cm} 
\def\xd{0.7cm}  
\def\ly{0.4cm}  
\def\yd{2*\ly+2*\ra+0.1cm} 

\newcommand{\circs}[2]{
	\begin{scope}[xshift=#1,yshift=#2,decoration={markings, mark=at position 0.59 with {\arrow{>}}}]
		\draw[\colfund,postaction={decorate}] (-\es,2*\ra)-- +(\es-\lx,0) arc (90:270:\ra) -- (\lx,0) arc (-90:90:\ra) --  (\es,2*\ra);
	\end{scope}
} 

\newcommand{\circss}[3]{
	\begin{scope}[xshift=#1,yshift=#2,decoration={markings, mark=at position 0.7 with {\arrow{>}}}]	
		\draw[\colfund,postaction={#3}] (-\es,2*\ra) -- +(\es-\lx,0) arc (90:270:\ra) --   (\xd+\lx,0) arc (-90:90:\ra) -- (\xd+\es,2*\ra);
		\draw[\colfund]   (\es,2*\ra) -- (\xd-\es,2*\ra);
	\end{scope}
} 

\newcommand{\circsss}[2]{
	\begin{scope}[xshift=#1,yshift=#2,decoration={markings, mark=at position 0.76 with {\arrow{>}}}]		
		\draw[\colfund,postaction={decorate}] (-\es,2*\ra) -- +(\es-\lx,0) arc (90:270:\ra) -- (2*\xd+\lx,0) arc (-90:90:\ra) -- (2*\xd+\es,2*\ra);
		\draw[\colfund]   (\es,2*\ra) -- (\xd-\es,2*\ra) (\xd+\es,2*\ra) -- (2*\xd-\es,2*\ra);
	\end{scope}
} 

\newcommand{\circssd}[2]{	
	\begin{scope}[xshift=#1,yshift=#2,decoration={markings, mark=at position 0.76 with {\arrow{>}}}]
		\draw[\colfund,postaction={decorate}] (-\es,2*\ra) -- +(\es-\lx,0) arc (90:270:\ra) --   (2*\xd+\lx,0) arc (-90:90:\ra) -- (2*\xd+\es,2*\ra);
		\draw[\colfund]   (\es,2*\ra) -- (2*\xd-\es,2*\ra);
	\end{scope}
} 

\newcommand{\circssu}[2]{
	\begin{scope}[xshift=#1,yshift=#2,decoration={markings, mark=at position 0.5 with {\arrow{>}}}]
		\draw[\colfund] (-\es,2*\ra) -- +(\es-\lx,0) arc (90:270:\ra) --   (\xd-\es,0) ;	
		\draw[\colfund,postaction={decorate}] (\xd+\es,0) -- (2*\xd+\lx,0) arc (-90:90:\ra) -- (2*\xd+\es,2*\ra);
		\draw[\colfund]   (\es,2*\ra) -- (\xd-\es,2*\ra) (\xd+\es,2*\ra) -- (2*\xd-\es,2*\ra);
	\end{scope}
} 

\newcommand{\strand}[2]{
	(#1,#2) -- ++(0,\yd)
} 

\newcommand{\cstrand}[2]{
	(#1,#2) -- ++(0,\ly) ++(0,0.1cm) -- ++(0,2*\ra + \ly)
} 

\newcommand{\cstrandb}[2]{
	(#1,#2) -- ++(0,\ly) ++(0,0.1cm) -- ++(0,2*\ra-0.1cm) ++(0,0.1cm) -- ++(0,\ly)
} 

\newcommand{\ucrossl}[2]{
	(#1,#2) -- ++(0,\ly+\ra - \xd/2+0.05cm) -- ++(\xd/2-0.05cm,\xd/2-0.05cm) ++(0.1cm,0.1cm) -- ++(\xd/2-0.05cm,\xd/2-0.05cm) -- ++(0,\ly+\ra - \xd/2+0.05cm) 	
} 

\newcommand{\ucrossr}[2]{
	(#1,#2) -- ++(0,\ly+\ra - \xd/2+0.05cm) -- ++(-\xd/2+0.05cm,\xd/2-0.05cm) -- ++(-0.1cm,0.1cm) -- ++(-\xd/2+0.05cm,\xd/2-0.05cm) -- ++(0,\ly+\ra - \xd/2+0.05cm)	
} 

\newcommand{\ocrossl}[2]{
	(#1,#2) -- ++(0,\ly+\ra - \xd/2+0.05cm) -- ++(\xd/2-0.05cm,\xd/2-0.05cm) -- ++(0.1cm,0.1cm) -- ++(\xd/2-0.05cm,\xd/2-0.05cm) -- ++(0,\ly+\ra - \xd/2+0.05cm) 	
} 

\newcommand{\ocrossr}[2]{
	(#1,#2) -- ++(0,\ly+\ra - \xd/2+0.05cm) -- ++(-\xd/2+0.05cm,\xd/2-0.05cm)  ++(-0.1cm,0.1cm) -- ++(-\xd/2+0.05cm,\xd/2-0.05cm) -- ++(0,\ly+\ra - \xd/2+0.05cm)	
} 

\def\rp{0.4cm} 
\def\rps{0.3cm}  
\def\rpb{0.32cm} 
\def\xdp{1cm}  
\def\sp{0.07cm} 

\newcommand{\punct}[2]{
	\node at (#1,#2) {\footnotesize $\bullet$};
} 

\newcommand{\punctn}[2]{
	\node at (#1,#2) {\footnotesize $\{{\color{\colfund} \bullet}\}$};
} 

\newcommand{\punctni}[3]{
	\node at (#1,#2) {\footnotesize $\{{\color{\colfund} \bullet}\}$};
	\node at (#1+0.05cm,#2-0.27cm) {\scriptsize $n_{#3}$};
} 

\newcommand{\pppunct}[2]{
	\punct{#1}{#2}
	\punct{#1+\xdp}{#2}
	\punct{#1+2*\xdp}{#2}
} 

\newcommand{\pppunctn}[2]{
	\punctn{#1}{#2}
	\punctn{#1+\xdp}{#2}
	\punctn{#1+2*\xdp}{#2}
} 

\newcommand{\circp}[3]{
		\draw[\colfund,decoration={markings, mark=at position 0 with {\arrow{>}}},postaction={#3}] (#1,#2) circle(\rp);
} 

\newcommand{\circps}[3]{
		\draw[\colfund,decoration={markings, mark=at position 0 with {\arrow{>}}},postaction={#3}] (#1,#2) circle(\rps);
} 

\newcommand{\circpp}[3]{
	\begin{scope}[xshift=#1,yshift=#2,decoration={markings, mark=at position 0.9 with {\arrow{>}}}]
		\draw[\colfund,postaction={#3}] (\xdp,-\rp) arc (-90:90:\rp) -- (0,\rp) arc (90:270:\rp) -- (\xdp,-\rp);
	\end{scope}
} 

\newcommand{\circpplb}[3]{
	\begin{scope}[xshift=#1,yshift=#2,decoration={markings, mark=at position 0.1 with {\arrow{>}}}]
		\draw[\colfund,postaction={#3}] (\rp-2*\sp,-\rp) -- (\xdp,-\rp) arc (-90:90:\rp) -- (\rp-2*\sp,\rp) (0,\rpb) arc (90:270:{\rp} and {\rpb});
	\end{scope}
} 

\newcommand{\circpprb}[3]{
	\begin{scope}[xshift=#1,yshift=#2,decoration={markings, mark=at position 0.6 with {\arrow{>}}}]
		\draw[\colfund,postaction={#3}] (\xdp-\rp+2*\sp,\rp) -- (0,\rp) arc (90:270:\rp) -- (\xdp-\rp+2*\sp,-\rp)  (\xdp,-\rpb) arc (-90:90:{\rp} and {\rpb});  
	\end{scope}
} 

\newcommand{\circppp}[3]{
	\begin{scope}[xshift=#1,yshift=#2,decoration={markings, mark=at position 0.95 with {\arrow{>}}}]
		\draw[\colfund,postaction={#3}] (2*\xdp,-\rp) arc (-90:90:\rp) -- (0,\rp) arc (90:270:\rp)  -- (2*\xdp,-\rp);
	\end{scope}
} 

\newcommand{\circppd}[3]{
	\begin{scope}[xshift=#1,yshift=#2,decoration={markings, mark=at position 0.405 with {\arrow{>}}}]
		\draw[\colfund,postaction={#3}]
		(\rp,0) arc (0:180:\rp) to [bend right=90] (2*\xdp+\rp,0) arc (0:180:\rp) to [bend left=90] (\rp,0);
	\end{scope}
} 

\newcommand{\circppu}[3]{
	\begin{scope}[xshift=#1,yshift=#2,decoration={markings, mark=at position 0.905 with {\arrow{>}}}]
		\draw[\colfund,postaction={#3}]
		(2*\xdp-\rp,0) arc (180:360:\rp) to [bend right=90] (-\rp,0) arc (180:360:\rp) to [bend left=90] (2*\xdp-\rp,0);
	\end{scope}
} 

\def\lxk{0.3cm} 
\def\sk{0.05cm} 
\def\bd{60} 

\title{\bf Chern--Simons theory, link invariants \\ and the Askey--Wilson algebra}
\renewcommand*{\Affilfont}{\normalsize\small}
\author[1]{Nicolas Cramp\'e}
\author[2]{Luc Vinet}
\author[3]{Meri Zaimi\vspace{.5em}}
\affil[1]{Institut Denis-Poisson CNRS/UMR 7013 - Université de Tours - Université
d'Orléans, \newline\vspace{.9em}
Parc de Grandmont, 37200 Tours, France.}
\affil[2,3]{Centre de Recherches Math\'ematiques, Universit\'e de Montr\'eal,
\newline\vspace{.9em}
P.O. Box 6128, Centre-ville Station, Montr\'eal (Qu\'ebec), H3C 3J7, Canada.}
\affil[2]{Insitut de valorisation des donn\'ees (IVADO), Montr\'eal (Qu\'ebec), H2S 3H1, Canada. \newline\vspace{.9em}}

{
 \makeatletter
 \renewcommand\AB@affilsepx{: \protect\Affilfont}
 \makeatother
 \affil[ ]{E-mail addresses}
 \makeatletter
 \renewcommand\AB@affilsepx{, \protect\Affilfont}
 \makeatother
 \affil[1]{crampe1977@gmail.com}
 \affil[2]{vinet@crm.umontreal.ca}
 \affil[3]{meri.zaimi@umontreal.ca}
 }
\begin{document}
	
\date{\today} 
\maketitle

\noindent{\bf Abstract:}
The occurrence of the Askey--Wilson (AW) algebra in the $SU(2)$ Chern--Simons (CS) theory and in the Reshetikhin--Turaev (RT) link invariant construction with quantum algebra $U_q(\mathfrak{su}_2)$ is explored. Tangle diagrams with three strands with some of them enclosed in a spin-$1/2$ closed loop are associated to the generators of the AW algebra. It is shown in both the CS theory and RT construction that the link invariant of these tangles obey the relations of the AW generators. It follows that the expectation values of certain Wilson loops in the CS theory satisfy relations dictated by the AW algebra and that the link invariants do not distinguish the corresponding linear combinations of links. 
\\[.5em]

\section{Introduction}

The purpose of this paper is to identify the presence of the Askey--Wilson algebra in the Chern--Simons theory and the Reshetikhin--Turaev link invariant construction, and to discuss the bearing it has in these contexts. 

One of the fundamental problems in knot theory is to determine whether two links in three-dimensional space are equivalent or not. In this regard, the study of link invariants plays an important role towards a classification of knots and links up to isotopy. A link invariant which is of interest in mathematical physics is the Jones polynomial, discovered in \cite{Jo85} via the study of a trace on the Temperley--Lieb (TL) algebra. This algebra, which was introduced in \cite{TL}, is connected to integrable lattice models in physics. Other examples of link invariants are the HOMFLY-PT \cite{HOMFLY,PT} and the Kauffman \cite{Kau90} polynomials, related respectively to the Hecke \cite{Jo87} and the Birman--Murakami--Wenzl (BMW) \cite{BW,Mur} algebras. These invariants are two-variable polynomials which contain the Jones polynomial as a special case. An important feature of the TL, Hecke and BMW algebras associated to these polynomial link invariants is that they are all quotients of the braid group algebra \cite{A}, which is central in the study of links.

The Chern--Simons (CS) theory is a quantum gauge field theory with an action that is defined on three-dimensional manifolds without use of a metric. For this reason, the theory is said to be topological. In \cite{Wit}, it is shown via the path-integral formalism that the expectation values of the observables of the theory, called the Wilson loops, lead to link invariants. In particular, the Jones polynomial is recovered when the manifold is the three-sphere, the gauge group is $SU(2)$ and all the Wilson loops are in the fundamental representation. The cases where the gauge group is $SU(N)$ or $SO(N)$ lead respectively to the HOMFLY-PT and the Kauffman polynomials (see also \cite{Ast,CGMM,GMM4,GMM5,Hor,KCP,WY} for instance). Hence, the CS theory provides an intrinsically three-dimensional interpretation of these link invariants, which usually require a two-dimensional projection of the link to be defined.

In \cite{Resh,Tur}, it is shown how link invariants can be constructed from Yang--Baxter representations of the braid group; this construction is related to how the Jones polynomial was originally obtained in \cite{Jo85}. Integrable systems satisfy an integrability condition called the Yang--Baxter equation. It is known that quasitriangular Hopf algebras yield interesting solutions of this equation through the $R$-matrix. The HOMFLY-PT and Kauffman polynomials are recovered in this construction \cite{Tur} when considering the $R$-matrices in the fundamental representation of the quantized universal enveloping algebras of $\su_N$ and $\mathfrak{so}_N$; the Jones polynomial is in particular associated to the quantum group $\Usl$. Such a mathematical framework for obtaining link invariants, to which we will refer as the Reshetikhin--Turaev (RT) construction, is further developed in terms of ribbon Hopf algebras and functors in \cite{RT}, with the aim of providing a mathematical realization of the CS quantum field approach of \cite{Wit}. Although the connections between the CS theory and the formalism of quantum groups and $R$-matrices have been investigated (see for instance \cite{GMM1,GMM2,GMM3,MS}), the equivalence of the link invariants obtained with both methods deserves to be spelled out.     

The Askey--Wilson (AW) algebra was first introduced in \cite{Zh}. It describes the bispectral properties of the Askey--Wilson polynomials, which form the family of basic hypergeometric orthogonal polynomials sitting on top of the $q$-Askey scheme \cite{Koek}. The AW algebra is in particular realized by the centralizer of the diagonal action of $\Usl$ in its threefold tensor product \cite{GZ,H}. The role of the $R$-matrix in this realization was showcased in \cite{CGVZ}. 
Due to this connection with the centralizer of $\Usl$, the overlap coefficients associated to the Racah problem for $\Usl$ are given in terms of the $q$-Racah polynomials, which are a finite truncation of the AW ones \cite{GZ,KR}. 
It is conjectured in \cite{CVZ2} that in general the centralizer of the diagonal action of $\Usl$ in the tensor product of three spin representations of $\Usl$ is isomorphic to a quotient of the AW algebra, in the spirit of a generalized Schur--Weyl duality. In particular, when considering three spins $1/2$ or three spins $1$, respectively, the TL and BMW algebras are recovered as quotients of the AW algebra. In the context of knot theory, the AW algebra is connected to the Kauffman skein algebra of (framed and unoriented) links in some punctured surfaces, see \cite{BP,Cooke,CFGPRV,Hik}.

The connections that exist between the link invariants mentioned above and various algebraic structures such as the braid group, the TL and BMW algebras, and the quantum group $\Usl$ all hint that the AW algebra must also belong to this picture. Therefore, it is natural to examine how the AW algebra features in the CS theory and related link invariant constructions. This is the goal of the present paper. Indeed, we show that the defining relations of the AW algebra appear in the study of the link invariants that arise from the CS theory with gauge group $SU(2)$ and the RT construction with quantum group $\Usl$ when considering any spin representations (see Theorems \ref{thm:CSAW} and \ref{thm:RTAW}). A consequence of these observations for the CS theory is that the expectation values of some products of Wilson loops are linearly related by the AW algebra. For knot theory, the implication is that the CS and RT link invariants do not distinguish between some linear combinations of links that correspond to the AW algebra relations. Moreover, the fact that the same results are obtained independently in the CS theory and in the RT mathematical framework of Yang--Baxter operators provides another reason to think that the link invariants in both cases are equivalent.   

The strategy we follow in this paper is to associate to the generators of the AW algebra some tangle diagrams composed of three straight strands with a subset of them being enclosed by a loop associated to the spin $1/2$ representation. These are inspired by the diagrams of the Kauffman skein algebra in \cite{CFGPRV}, where punctures on a plane are enclosed by loops. Then, we show that the values of the link invariants of these diagrams satisfy the defining relations of the AW algebra both (and independently) in the CS theory and the RT construction. In the CS case, the proof is done by using the connections with the Kauffman bracket polynomial while in the RT case, it is done by computing algebraically some partial traces of $R$-matrices and identifying them as the intermediate Casimir elements of $\Ut$.   

The paper is organized as follows. Section \ref{sec:KnotTheory} contains a review of relevant concepts of the theory of knots and links (Subsection \ref{ssec:KnotsLinks}) as well as a definition of the braid group and its associated braid diagrams (Subsection \ref{ssec:BraidGroup}). In Section \ref{sec:AWdiag}, the AW tangle diagrams are defined and put in correspondence with the generators of the AW algebra. Section \ref{sec:AWCS} provides a proof that these AW diagrams lead to the AW relations in the CS theory on $\bR^3$ with gauge group $SU(2)$. Some preliminaries on the CS theory are first recalled in Subsection \ref{ssec:CSWL}, and the relevant properties of the Wilson loop expectation values are presented in Subsection \ref{ssec: ProprWL}. The known connection between the CS link invariant and Kauffman's bracket polynomial is given in Subsection \ref{ssec:KauffmanBracket}. Then, the properties of the CS link invariant are used in Subsection \ref{ssec:AWWL} to prove that the AW diagrams satisfy the AW relations. Subsection \ref{ssec:TL} comments on the connection with the Temperley--Lieb algebra. In Section \ref{sec:AWRT}, it is shown that the same AW tangle diagrams also lead to the AW relations in the RT construction of link invariants associated to $\Usl$. Subsection \ref{ssec:Uqsu2} first recalls the definition and properties of the quantum group $\Usl$, and Subsection \ref{ssec:RepU} briefly discusses its finite irreducible representations. The RT construction of link invariants via traces of Yang--Baxter operators is explained in Subsection \ref{ssec:TrLinkInv}. Then, it is shown in Subsection \ref{ssec:CasimirAW} that the AW tangle diagrams, viewed as partially closed braids, correspond to partial traces of $R$-matrices that are equal to the intermediate Casimir elements of $\Ut$. Section \ref{sec:concl} contains concluding remarks. The paper is complemented by two appendices. The first (Appendix \ref{app:proofPropTrEnh}) contains the technical proof of a proposition regarding a known property of the trace in the RT construction. The second (Appendix \ref{app:AWRmatrix}) provides a new proof that the intermediate Casimir elements of $\Ut$ satisfy the relations of the AW algebra, using the formalism of $R$-matrices and their partial traces.           

\section{Knot theory and braid group}\label{sec:KnotTheory}

This section recalls definitions and properties regarding knots, links and braids.  

\subsection{Knots and links}\label{ssec:KnotsLinks}

A knot is a smooth embedding of the circle $S^1$ in $\bR^3$. A link with $n$ components is the union of $n$ knots that do not intersect. In this paper, we will consider oriented knots and links, unless stated otherwise.

Two links in $\bR^3$ are said to be ambiant isotopic if one can be smoothly deformed into the other in $\bR^3$. This defines an equivalence relation for links in $\bR^3$. A link invariant is a mapping $L \mapsto I(L)$ such that $I(L_1)=I(L_2)$ if $L_1$ and $L_2$ are equivalent links.

A link $L$ in $\bR^3$ is conveniently represented by a link diagram $D_L$, which is a projection of $L$ on a two-dimensional plane $\bR^2$ with a finite number of crossings. We distinguish between the overcrossing ($L_+$) and undercrossing ($L_-$) configurations illustrated in Figure \ref{fig:crossing}.  
\begin{figure}[h]
	\begin{tikzpicture}[scale=1,line width=\lw]
		\draw[rounded corners=5pt,->] \ocrossl{0}{0};
		\draw[rounded corners=5pt,->] \ocrossr{\xd}{0};
		\node at (\xd/2,-0.3cm) {$L_+$};
	\end{tikzpicture} \hspace{2cm}
	\begin{tikzpicture}[scale=1,line width=\lw]
		\draw[rounded corners=5pt,->] \ucrossl{0}{0};
		\draw[rounded corners=5pt,->] \ucrossr{\xd}{0};
		\node at (\xd/2,-0.3cm) {$L_-$};
	\end{tikzpicture} \hspace{2cm}
	\begin{tikzpicture}[scale=1,line width=\lw]
		\draw[->] \strand{0}{0};
		\draw[->] \strand{\xd}{0};
		\node at (\xd/2,-0.3cm) {$L_0$};
	\end{tikzpicture}
	\centering
	\caption{Crossing configurations.}
	\label{fig:crossing}
\end{figure}
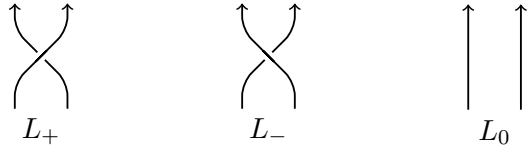
A crossing of type $L_+$ (resp. $L_-$) is said to have positive (resp. negative) sign. The writhe number $w(D_L)$ of a link diagram $D_L$ is defined as the sum of the signs of all the crossings. We also define the configuration $L_0$ illustrated in Figure \ref{fig:crossing}, which has no crossing.
 
Two link diagrams represent ambiant isotopic links if and only if they are related by a finite sequence of planar isotopies and Reidemeister moves (RM), illustrated in Figure \ref{fig:RM}. The equivalence relation for link diagrams induced by planar isotopies and RM of type II and III only is referred to as regular isotopy \cite{Kau90}.
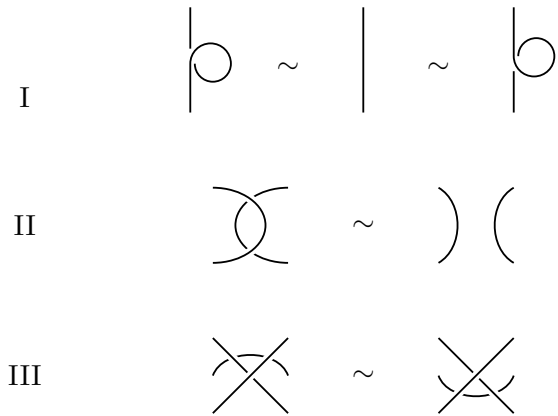
\begin{figure}[h]
	\begin{tikzpicture}[scale=1,line width=\lw]
		\node at (0,0.2cm) {I};
		\begin{scope}[xshift=2.2cm]
			\draw[] (0,0) -- (0,\ly+\ra) arc(180:0:0.27cm) arc(0:-180: 0.24cm)  (0,\ly+\ra+0.2cm) -- (0,\yd);
			\node at (1.3cm,0.6cm) {$\sim$};
			\draw[] (2.3cm,0) -- +(0,\yd);
			\node at (3.3cm,0.6cm) {$\sim$};
			\draw[xshift=4.3cm] (0,\yd) -- (0,\yd-\ly-\ra) arc(-180:0:0.27cm) arc(0:180: 0.24cm) (0,\yd-\ly-\ra-0.2cm) -- (0,0);	
		\end{scope}
		
		\node at (0,-1.5cm) {II};
		\begin{scope}[xshift=2.5cm,yshift=-2cm]
			\draw[white,line width=1.3pt,double=black,double distance=\lw] (1cm,0) arc(-90:-270:0.7cm and 0.5cm) +(0,1cm);
			\draw[white,line width=1.3pt,double=black,double distance=\lw] (0,0) arc(-90:90:0.7cm and 0.5cm) +(0,1cm);
			\node at (2cm,0.5cm) {$\sim$};
			\draw[] (3cm,0) to [bend right=60] +(0,1cm);
			\draw[] (4cm,0) to [bend left=60] +(0,1cm);
		\end{scope}
		
		\node at (0,-3.5cm) {III};
		\begin{scope}[xshift=2.5cm,yshift=-4cm]
			\draw[white,line width=1.3pt,double=black,double distance=\lw] (0,0.5cm) to [bend left=70] (1cm,0.5cm);
			\draw[white,line width=1.3pt,double=black,double distance=\lw] (1cm,0) -- (0,1cm);
			\draw[white,line width=1.3pt,double=black,double distance=\lw] (0,0) -- (1cm,1cm);
			\node at (2cm,0.5cm) {$\sim$};
			\draw[white,line width=1.3pt,double=black,double distance=\lw] (3cm,0.5cm) to [bend right=70] (4cm,0.5cm);
			\draw[white,line width=1.3pt,double=black,double distance=\lw] (4cm,0) -- (3cm,1cm);
			\draw[white,line width=1.3pt,double=black,double distance=\lw] (3cm,0) -- (4cm,1cm);
		\end{scope}
	\end{tikzpicture}
	\centering
	\caption{The Reidemeister moves (where the components can have any orientation).}
	\label{fig:RM}
\end{figure}

A framing of a link $L$ is a continuous and nowhere vanishing vector field which is normal to $L$. Therefore, framed links can be viewed as bands. Instead of representing framed links by diagrams of bands, we choose to represent them by usual link diagrams with the convention that, for each component of a link, the framing is given by a normal vector field which is always perpendicular to the projection plane. This is known as the vertical framing. In this representation, the configurations illustrated in Figure \ref{fig:onecompconf} are not equivalent. Therefore, in the vertical framing convention, the RM of type I is not valid anymore. Note however that the RM of types II and III can still be used. For this reason, regular isotopy is relevant when studying framed links.        

\begin{figure}[h]
	\begin{tikzpicture}[scale=1,line width=\lw]
		\draw[] (0,0) -- (0,\ly+\ra) arc(180:0:0.27cm) arc(0:-180: 0.24cm)  (0,\ly+\ra+0.2cm) -- (0,\yd);
		\node at (0.2cm,-0.3cm) {$L^{(+)}$};
	\end{tikzpicture} \hspace{2cm}
	\begin{tikzpicture}[scale=1,line width=\lw]
		\draw[] (0,\yd) -- (0,\yd-\ly-\ra) arc(-180:0:0.27cm) arc(0:180: 0.24cm);  
		\draw (0,\yd-\ly-\ra-0.2cm) -- (0,0);
		\node at (0.2cm,-0.3cm) {$L^{(-)}$};
	\end{tikzpicture} \hspace{2cm}
	\begin{tikzpicture}[scale=1,line width=\lw]
		\draw[] \strand{0}{0}; 
		\node at (0.2cm,-0.3cm) {$L^{(0)}$};
	\end{tikzpicture}
	\centering
	\caption{Three configurations which are not equivalent for framed links.}
	\label{fig:onecompconf}
\end{figure}
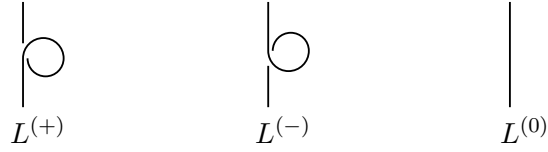

Finally, a link $L$ with $n$ components is said to be colored if each component is associated to some parameter $\alpha_i$ (referred to as the ``color''), for $i=1,...,n$. 

In this paper, we will be interested by framed and colored links, where the colors are non-negative integers or half-integers to be called ``spins''.

\subsection{Braid group}\label{ssec:BraidGroup}

A convenient way of studying links is via braids. The braid group on $n$ strands $B_n$ is generated by invertible elements $\sigma_i$, for $i=1,...,n-1$, which satisfy the following defining relations
\begin{align}
	&\sigma_i\sigma_j=\sigma_j\sigma_i,& &|i-j|>1, \label{eq:relBraid1}\\
	&\sigma_{i} \sigma_{i+1} \sigma_{i} = \sigma_{i+1} \sigma_{i} \sigma_{i+1},& &i=1,...,n-2. \label{eq:relBraid2}
\end{align}

The generators and their inverses can equivalently be seen as the following braid diagrams
\begin{equation}
	\sigma_i =
	\begin{tikzpicture}[scale=1,line width=\lw,baseline={([yshift=-\eseq+0.3cm]current bounding box.center)}]
		\draw[->] \strand{0}{0};
		\node at (0,-0.3cm) {\footnotesize $1$};
		\node at (0.35cm,0) {...};
		\draw[rounded corners=5pt,->] \ocrossl{\xd}{0};
		\node at (\xd,-0.3cm) {\footnotesize $i$};
		\draw[rounded corners=5pt,->] \ocrossr{2*\xd}{0};
		\node at (2*\xd,-0.3cm) {\footnotesize $i+1$};
		\node at (2*\xd+0.35cm,0) {...};
		\draw[->] \strand{3*\xd}{0};
		\node at (3*\xd,-0.3cm) {\footnotesize $n$};
	\end{tikzpicture} \ , \qquad
	\sigma_i^{-1} =
	\begin{tikzpicture}[scale=1,line width=\lw,baseline={([yshift=-\eseq+0.3cm]current bounding box.center)}]
		\draw[->] \strand{0}{0};
		\node at (0,-0.3cm) {\footnotesize $1$};
		\node at (0.35cm,0) {...};
		\draw[rounded corners=5pt,->] \ucrossl{\xd}{0};
		\node at (\xd,-0.3cm) {\footnotesize $i$};
		\draw[rounded corners=5pt,->] \ucrossr{2*\xd}{0};
		\node at (2*\xd,-0.3cm) {\footnotesize $i+1$};
		\node at (2*\xd+0.35cm,0) {...};
		\draw[->] \strand{3*\xd}{0};
		\node at (3*\xd,-0.3cm) {\footnotesize $n$};
	\end{tikzpicture} \ ,
\end{equation}
with the group product given by vertical concatenation. A general braid diagram is such that $n$ points on a
line are connected by always upgoing strings to $n$ points on a line above, with an overcrossing and undercrossing specification. The defining relations of $B_n$ are seen to express an isotopy equivalence of braids. In particular, the relation $\sigma_i\sigma_i^{-1}=1$ corresponds to the RM of type II while the relation \eqref{eq:relBraid2} corresponds to the RM of type III.

The importance of braids for links is given by Alexander's theorem stating that any link can be represented as the closure of a braid, obtained by connecting the top ends of the strings on a braid diagram with their corresponding bottom ends.

\section{Askey--Wilson diagrams}\label{sec:AWdiag}

In this section, we define the diagrams associated to the generators of the Askey--Wilson algebra, as will be shown later. 

The idea is to consider three colored vertical strands and to encircle a (non-empty) subset of them by a loop with spin $1/2$. Since the spin $1/2$ will play a special role in what is to come, we will represent any component of a diagram associated to this spin in blue. Hence, we define
\begin{gather}
	\A_1 := \
	\begin{tikzpicture}[scale=1,line width=\lw,baseline={([yshift=-\eseq]current bounding box.center)}]
		\draw[->] \cstrand{0}{0};
		\circs{0}{\ly+0.05cm}
		\draw[->] \strand{\xd}{0};
		\draw[->] \strand{2*\xd}{0};
	\end{tikzpicture} \ , \qquad
	\A_2 := \
	\begin{tikzpicture}[scale=1,line width=\lw,baseline={([yshift=-\eseq]current bounding box.center)}]
		\draw[->] \strand{0}{0};
		\draw[->] \cstrand{\xd}{0};
		\circs{\xd}{\ly+0.05cm}
		\draw[->] \strand{2*\xd}{0};
	\end{tikzpicture}  \qquad
	\A_3 := \
	\begin{tikzpicture}[scale=1,line width=\lw,baseline={([yshift=-\eseq]current bounding box.center)}]
		\draw[->] \strand{0}{0};
		\draw[->] \strand{\xd}{0};
		\draw[->] \cstrand{2*\xd}{0};
		\circs{2*\xd}{\ly+0.05cm}
	\end{tikzpicture} \ , \label{eq:Ci} \\
	\A_{12} := \
	\begin{tikzpicture}[scale=1,line width=\lw,baseline={([yshift=-\eseq]current bounding box.center)}]
		\draw[->] \cstrand{0}{0};
		\draw[->] \cstrand{\xd}{0};
		\circss{0}{\ly+0.05cm}{decorate}
		\draw[->] \strand{2*\xd}{0};
	\end{tikzpicture} \ , \qquad 
	\A_{23} := \
	\begin{tikzpicture}[scale=1,line width=\lw,baseline={([yshift=-\eseq]current bounding box.center)}]
		\draw[->] \strand{0}{0};
		\draw[->] \cstrand{\xd}{0};
		\draw[->] \cstrand{2*\xd}{0};
		\circss{\xd}{\ly+0.05cm}{decorate}
	\end{tikzpicture} \ , \qquad
	\A_{13} := \
	\begin{tikzpicture}[scale=1,line width=\lw,baseline={([yshift=-\eseq]current bounding box.center)}]
		\draw[->] \cstrand{0}{0};
		\draw[->] \strand{\xd}{0};
		\draw[->] \cstrand{2*\xd}{0}; 
		\circssu{0}{\ly+0.05cm}
	\end{tikzpicture} \ , \label{eq:Cij} \\
	\A_{123} := \
	\begin{tikzpicture}[scale=1,line width=\lw,baseline={([yshift=-\eseq]current bounding box.center)}]
		\draw[->] \cstrand{0}{0};
		\draw[->] \cstrand{\xd}{0};
		\draw[->] \cstrand{2*\xd}{0};
		\circsss{0}{\ly+0.05cm}
	\end{tikzpicture} \ . \label{eq:C123}
\end{gather}
The previous diagrams are known as colored and oriented (3,3)-tangle diagrams in knot theory\footnote{A $(m,n)$-tangle diagram is such that $m$ points on a bottom line and $n$ points on a top line are connected by arcs, and it can contain loops between the two lines.}. They are similar to braid diagrams on three strands, with the difference that they contain loops. These tangle diagrams can be viewed to represent part of a colored and oriented link, where the colors are given by spins. Unless required, the colors of the vertical strands are not indicated on the figure. Since we will consider framed links, the equivalence relation which will be of interest for us is regular isotopy. Hence, if $X$ and $Y$ are two regular isotopic tangle diagrams, we will write $X=Y$ to refer to their equivalence.  

In what follows, we will be interested in deriving relations between the values of invariants associated to framed links which only differ in some finite region by a combination of the tangle diagrams \eqref{eq:Ci}--\eqref{eq:C123}. To do so, we first define the product $XY$ of two such tangle diagrams $X$ and $Y$ by vertical concatenation, with the convention that the diagram $X$ is put on top of the diagram $Y$. This product rule together with the regular isotopy equivalence imply that the diagrams $\A_1,\A_2,\A_3$ and $\A_{123}$ commute with all the diagrams in \eqref{eq:Ci}--\eqref{eq:C123}. Indeed, one can move the loops enclosing only one strand inside the larger loops, and the loops enclosing all three strands around the smaller loops.  However, the diagrams $\A_{12}$ and $\A_{23}$ do not commute {\it a priori} because one cannot pass the top blue loop below the bottom one without encountering a crossing in the following diagrams:
\begin{equation}
	\A_{12}\A_{23} = \
	\begin{tikzpicture}[scale=1,line width=\lw,baseline={([yshift=-\eseq]current bounding box.center)}]
		\draw[->] \strand{0}{0} -- ++\cstrand{0}{0};
		\draw[->] \cstrand{\xd}{0} -- ++\cstrand{0}{0};
		\draw[->] \cstrand{2*\xd}{0} -- ++\strand{0}{0};
		\circss{\xd}{\ly+0.05cm}{decorate}
		\circss{0}{\ly+0.05cm+\yd}{decorate}
	\end{tikzpicture} \ , \qquad 
	\A_{23}\A_{12} = \
	\begin{tikzpicture}[scale=1,line width=\lw,baseline={([yshift=-\eseq]current bounding box.center)}]
		\draw[->] \cstrand{0}{0} -- ++\strand{0}{0};
		\draw[->] \cstrand{\xd}{0} -- ++\cstrand{0}{0};
		\draw[->] \strand{2*\xd}{0} -- ++\cstrand{0}{0};
		\circss{0}{\ly+0.05cm}{decorate}
		\circss{\xd}{\ly+0.05cm+\yd}{decorate}
	\end{tikzpicture} \ . \label{eq:prodC12C23}
\end{equation}

We can also define similarly the product of a braid $\sigma \in B_3$ with the $(3,3)$-tangle diagrams in \eqref{eq:Ci}--\eqref{eq:C123}. It is then seen (by using again regular isotopy invariance) that the element $\A_{13}$ can be expressed in terms of $\A_{12}$ or $\A_{23}$ as follows:
\begin{equation}
	\A_{13} = \sigma_1 \A_{23} \sigma_1^{-1} = \
	\begin{tikzpicture}[scale=1,line width=\lw,baseline={([yshift=-\eseq]current bounding box.center)}]
		\draw[rounded corners=5pt,->] \ucrossl{0}{0} -- ++\cstrand{0}{0} ++\ocrossr{0}{0};
		\draw[rounded corners=5pt,->] \ucrossr{\xd}{0} -- ++\strand{0}{0} ++\ocrossl{0}{0};
		\circss{\xd}{\ly+0.05cm+\yd}{decorate}
		\draw[->] \strand{2*\xd}{0} -- ++\cstrand{0}{0} -- ++\strand{0}{0};
	\end{tikzpicture} \ = \
	\begin{tikzpicture}[scale=1,line width=\lw,baseline={([yshift=-\eseq]current bounding box.center)}]
		\draw[->] \strand{0}{0} -- ++\cstrand{0}{0} -- ++\strand{0}{0};
		\draw[rounded corners=5pt,->] \ocrossl{\xd}{0} -- ++\strand{0}{0} ++\ucrossr{0}{0};
		\draw[rounded corners=5pt,->] \ocrossr{2*\xd}{0} -- ++\cstrand{0}{0} ++\ucrossl{0}{0};
		\circss{0}{\ly+0.05cm+\yd}{decorate}
	\end{tikzpicture} \ =
	\sigma_2^{-1} \A_{12} \sigma_2. \label{eq:C130sig}
\end{equation}

In order to state the connection between the diagrams \eqref{eq:Ci}--\eqref{eq:C123} and the Askey--Wilson algebra, we need the following definition.
\begin{defi}\label{def:AW} \cite{CFGPRV}
	The special Askey--Wilson algebra $\saw(3)$ is generated by $\C_{12},\C_{23},\C_{13}$ and central elements $\C_1,\C_2,\C_3,\C_{123}$ subject to the relations
	\begin{align}
		&[\C_{12},\C_{23}]_q + (q^2-q^{-2})\C_{13} =(q-q^{-1})(\C_1\C_3+\C_2 \C_{123}), \label{eq:AW1} \\
		&[\C_{23},\C_{13}]_q + (q^2-q^{-2})\C_{12} =(q-q^{-1})(\C_1\C_2+\C_3 \C_{123}), \label{eq:AW2} \\
		&[\C_{13},\C_{12}]_q + (q^2-q^{-2})\C_{23} =(q-q^{-1})(\C_2\C_3+\C_1 \C_{123}), \label{eq:AW3}\\
		&q \C_{12}\C_{23}\C_{13} + q^2 \C_{12}^2 + q^{-2} \C_{23}^2 + q^2 \C_{13}^2 -q\C_{12} (\C_1\C_2+\C_3 \C_{123})- q^{-1} \C_{23} (\C_2\C_3+\C_1 \C_{123}) \nonumber \\
		& - q \C_{13}(\C_1\C_3+\C_2 \C_{123}) = (q+\qi)^2 - \C_{123}^2 -\C_1^2 -\C_2^2 - \C_3^2 - \C_1\C_2\C_3\C_{123}, \label{eq:AW4}		
	\end{align}
	where $q$ is a complex number and $[X,Y]_q=qXY-\qi YX$ is the $q$-commutator.
\end{defi}
Note that \eqref{eq:AW1}--\eqref{eq:AW3} are the defining relations of (a centrally extended version of) the original AW algebra introduced in \cite{Zh}. Moreover, the LHS of \eqref{eq:AW4} is a Casimir element for this algebra.
  
The main goal of this paper is to show that, when considering some specific link invariants, the tangle diagrams \eqref{eq:Ci}--\eqref{eq:C123} obey the Askey--Wilson algebra relations \eqref{eq:AW1}--\eqref{eq:AW3} under the correspondence
\begin{equation}
	\C_I \mapsto \A_I, \qquad \forall I \in \{1,2,3,12,23,13,123\}. \label{eq:corres}
\end{equation}
Let us mention that this correspondence between the AW algebra generators $\C_I$ and the tangle diagrams $\A_I$ together with equation \eqref{eq:C130sig} is consistent with our understanding of the realization of the AW algebra as the centralizer of $\Usl$ in $\Ut$ \cite{CGVZ}. We will come back to this point later in Subsection \ref{ssec:CasimirAW}. 

The two link invariants which will be considered here are: the one which arises from the Chern--Simons quantum field theory with gauge group $SU(2)$, and the one which follows from the Reshetikhin--Turaev construction associated to the quantum group $\Usl$. Let us mention again that there are reasons to believe that these two link invariants are in fact the same (see \cite{Gua,GMM1,GMM2,GMM3,MS,RT} for instance), but since this connection is not obvious to establish, we will consider both cases independently.

\section{Askey--Wilson algebra in the Chern--Simons theory}\label{sec:AWCS}
This section focuses on the Chern--Simons quantum field theory and its link invariants. It will be shown that the tangle diagrams defined in \eqref{eq:Ci}-\eqref{eq:C123} lead indeed to the Askey--Wilson algebra in this context.

\subsection{Chern--Simons action and Wilson loops}\label{ssec:CSWL}

Throughout this paper, we consider the Chern--Simons theory on $\mathbb{R}^3$ with gauge group $SU(2)$. We will mainly follow the conventions of \cite{Gua}. The Lie algebra $\su_2$ has generators $T^a$, for $a=1,2,3$, with Lie bracket $[T^a,T^b]=i\epsilon^{abc} T^c$, where $\epsilon^{abc}$ is the Levi-Civita symbol. For each spin $j=0,1/2,1,...$, $\su_2$ has an irreducible representation of finite dimension $2j+1$. In the fundamental (spin-$1/2$) representation, the generators $T^a$ are represented by the Pauli matrices with a normalization factor of $1/2$. The gauge potential of the field theory is the one-form $A=\sum_{\mu} A_\mu dx^\mu$ for $\mu=0,1,2$ with values in $\su_2$: $A_\mu(x) = \sum_{a}A^a_\mu(x) T^a$.

The Chern--Simons action is
\begin{equation}
	S_{CS}=\frac{\kappa}{4\pi}\int_{\bR^3} \Tr\left( A \wedge dA + \frac{2i}{3} A \wedge A \wedge A \right), \label{eq:SCS}
\end{equation} 
where $\kappa$ is the coupling constant, $\wedge$ is the exterior product and $\Tr$ is the trace in the two-dimensional spin-$1/2$ representation of $\su_2$. The action $S_{CS}$ as defined in \eqref{eq:SCS} is manifestly a topological invariant since it is the integral of a three-form over a three-manifold. Moreover, it is invariant under the gauge transformation
\begin{equation}
	A_\mu(x) \to A_\mu^\Omega(x) = \Omega^{-1}(x) A_\mu(x) \Omega(x) - i \Omega^{-1}(x) \partial_\mu \Omega(x), \label{eq:gaugeTransf}
\end{equation}
where $\Omega:\bR^3 \to SU(2)$ is a smooth map.

The gauge invariant observables of the CS theory are the Wilson loops, defined by 
\begin{equation}
	W(\gamma,j)=\Tr\left[ P \exp\left(i \oint_{\gamma} A_\mu^a T^a_{(j)} dx^\mu \right)\right], \label{eq:WL}
\end{equation}  
where $\gamma$ is a closed and oriented smooth curve in $\bR^3$, $P$ is the path-ordering operator and $T^a_{(j)}$ denotes the spin-$j$ representation of the generator $T^a$ of $\su(2)$. (Note that one speaks of ``Wilson lines'' if the integral in \eqref{eq:WL} is taken along a path which is not closed and if there is no trace.) More generally, we can consider the finite union of non-intersecting closed and oriented smooth curves $\gamma_i$ each associated to some spin $j_i$; this defines an oriented and colored link $L$. Then the product of Wilson loops associated to this link $L$ is
\begin{equation}
	W(L)=W(\gamma_1,...,\gamma_n;j_1,...,j_n)=\prod_{i=1}^n W(\gamma_i,j_i). \label{eq:prodWL}
\end{equation}
The vacuum expectation value of a product of Wilson loops is given in terms of path integrals by
\begin{equation}
	\left\langle W(L) \right\rangle= \frac{\int \mathcal{D}A \  W(L) e^{iS_{CS}} }{\int \mathcal{D}A \  e^{iS_{CS}}}. \label{eq:vevWL}
\end{equation}
In order for this expression to be well-defined, the link $L$ must be framed. Hence, the expectation value \eqref{eq:vevWL} depends on a choice of framing. To specify the choice of framing on link diagrams, we will use the vertical framing (VF) convention. Therefore, the object of interest will be
\begin{equation}
	\E(L) := \left\langle W(L) \right\rangle_{\text{VF}},  \label{eq:vevWLVF}
\end{equation}
where we use the same notation $L$ for a link in $\bR^3$ and its two-dimensional projection diagram. Note that $\E(L)$ can be expressed in terms of the deformation parameter
\begin{equation}
	q:=\exp\left(-\frac{i \pi }{\kappa}\right). \label{eq:defq}
\end{equation}
In \cite{Gua}, the object $\E(L)$ is denoted $E(L)$ and the deformation parameter corresponds to $q^2$.

\subsection{Properties of the Wilson loop expectation values} \label{ssec: ProprWL}

We now provide a list of properties of $\E(L)$ which will be useful for what follows and which can be derived from the CS theory \cite{Gua,Wit}.
\begin{enumerate}
	\item $\E$ is an ambiant isotopy invariant of oriented, colored and framed links in $\bR^3$, and a regular isotopy invariant of oriented and colored link diagrams. This is a consequence of the topological invariance of the CS theory. Actually, since we are considering the gauge group $SU(2)$, the value $\E(L)$ does not depend on the orientation of the components of the link $L$, hence the previous statement holds for unoriented links (see for instance \cite{Gua} for more details).
	\item If two links $L^{(\pm)}$ and $L^{(0)}$ are identical everywhere except at some small region where they look as indicated in Figure \ref{fig:onecompconf}, and if the component which differs is associated to the spin $j$, then
	\begin{equation}
		\E(L^{(\pm)};j) = q^{\pm 2j(j+1)} \E(L^{(0)};j). \label{eq:ERMI}
	\end{equation}   
	The property \eqref{eq:ERMI} corresponds to a change of framing of one of the components of the link.
	\item If $L_1$ and $L_2$ are two disjoint links, then
	\begin{equation}
		\E(L_1 \cup L_2) = \E(L_1)\E(L_2). \label{eq:Efact}
	\end{equation}
	This factorization property is a consequence of the topological invariance of the CS theory and of the uniqueness of the vacuum.
	\item If $U_1$ and $U_2$ are two unknots with zero writhe (see Figure \ref{fig:unknot}) and respective spins $j_1$ and $j_2$, then
	\begin{equation}
		\E(U_1,U_2;j_1,j_2)=\sum_{j=|j_1-j_2|}^{j_1+j_2} \E(U_1;j). \label{eq:Efus}
	\end{equation}
	This is a consequence of the fusion property of the Wilson loops and of the direct sum decomposition rule for the tensor product of two spin representations of $\su_2$.
	\item In a $1+2$ time and space decomposition of $\bR^3$, the (full) monodromy matrix which describes the braiding of two Wilson lines associated to the representations of spins $j_1$ and $j_2$ of $\su_2$ is given by
	\begin{equation}
		\cB = q^{4 T^a_{(j_1)} \otimes T^a_{(j_2)}}. \label{eq:monodromyMat}
	\end{equation}
	The eigenvalues of the monodromy matrix are
	\begin{equation}
		q^{-2j_1(j_1+1)-2j_2(j_2+1)+2j(j+1)}, \label{eq:monodromyEV}
	\end{equation}
	for $j = |j_1-j_2|,|j_1-j_2|+1,...,j_1+j_2$. The half-monodromy matrix is given by
	\begin{equation}
		\mathcal{M}=\Pi_{12}q^{2 T^a_{(j_1)} \otimes T^a_{(j_2)}}, \label{eq:halfmonodromyMat}
	\end{equation}
	where $\Pi_{12}$ is the permutation operator that exchanges the representation spaces $j_1$ and $j_2$. This matrix is associated to the exchange of two punctures in the 2-space of the CS theory. In terms of link diagrams, it corresponds to performing the crossing $L_+$ of Figure \ref{fig:crossing} (and the inverse matrix corresponds to $L_-$). The minimal characteristic polynomial of the half-monodromy matrix leads to a skein relation (or a generalized version) for the CS link invariant. For instance, when two Wilson lines are in the spin $1/2$ representation, $\mathcal{M}$ has two eigenvalues and one gets the relation
	\begin{equation}
		q^{\frac{1}{2}}\E(L_+) - q^{-\frac{1}{2}}\E(L_-) = (q-\qi) \E(L_0). \label{eq:skein} 
	\end{equation} 
\end{enumerate}
These properties of the Wilson loop expectation values allow one to compute the CS invariant associated to any link in $\bR^3$ (see \cite{Gua}). The result is a Laurent polynomial in the variable $q^{\frac{1}{2}}$. Some useful cases are given in the following proposition.
\begin{prop} \cite{Gua}
	Denote $[x]_q:=\frac{q^x-q^{-x}}{q-\qi}$.
	\begin{enumerate}[label=(\alph*)]
		\item The value of the unknot $U$ with zero writhe (see Figure \ref{fig:unknot}) associated to the irreducible representation of spin $j$ of $\su_2$ is 
		\begin{equation}
			\E(U;j) = [2j+1]_q. \label{eq:Eunknot}
		\end{equation}
		\item The value of the Hopf link $L_H$ (see Figure \ref{fig:hopf}) whose components have zero writhe and are associated to the irreducible representations of spins $j_1$ and $j_2$ of $\su_2$ is
		\begin{equation}
			\E(L_H;j_1,j_2) = [(2j_1+1)(2j_2+1)]_q. \label{eq:EHopf}
		\end{equation}  
	\end{enumerate}
\end{prop}
\begin{figure}[h]
	\hspace*{\fill}%
	\begin{subfigure}[b]{0.1\textwidth}
		\centering
		\begin{tikzpicture}[scale=1,line width=\lw]
			\draw[->] (0.5cm,0) arc(0:360:0.6cm);
		\end{tikzpicture}
		\caption{}
		\label{fig:unknot}
	\end{subfigure}
	\hspace*{\fill}%
	\begin{subfigure}[b]{0.1\textwidth}
		\centering
		\begin{tikzpicture}[scale=1,line width=\lw]
			\draw[decoration={markings, mark=at position 0.85 with {\arrow{>}}},postaction=decorate] (70:0.6cm) arc(70:410:0.6cm);
			\draw[xshift=0.6cm,decoration={markings, mark=at position 0.875 with {\arrow{<}}},postaction=decorate] (-110:0.6cm) arc(-110:230:0.6cm);
		\end{tikzpicture}
		\caption{}
		\label{fig:hopf}
	\end{subfigure}
	\hspace*{\fill}%
	\caption{The unknot (a) and the Hopf link (b).}
	\label{fig:UHopf}
\end{figure}
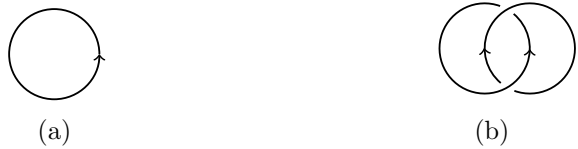
When all the Wilson loops are in the spin-$1/2$ representation, the regular isotopy invariance, the value of the unknot \eqref{eq:Eunknot}, the property \eqref{eq:ERMI} and the skein relation \eqref{eq:skein} uniquely determine the value $\E(L)$ (see \cite{Gua}). More generally, the CS link invariant can be determined by a recursive method which reduces the computation to the case where all the spins are $1/2$, as per the following proposition. 
\begin{prop} \label{prop:Erec} \cite{Gua} Let $L$ be a link with components $K_1,K_2,...$ associated to spins $j_1,j_2,...$ 
	\begin{enumerate}[label=(\alph*)]
		\item Any component $K_i$ in the trivial representation can be removed from the computation of $\E(L)$. That is, if $j_1=0$ (without loss of generality), then
		\begin{equation}
			\E(K_1,K_2,...;0,j_2,...) = \E(K_2,...;j_2,...). \label{eq:Ereptriv}
		\end{equation}
		\item For any spin $j_1$ (and similarly for any other spin $j_i$ of $L$),
		\begin{align}
			\E(K_1,K_2,...;j_1,j_2,...) =& \E(K_1,K_1,K_2,...;1/2,j_1-1/2,j_2,...) \label{eq:Erec} \\
			&- \E(K_1,K_2,...;j_1-1,j_2,...), \nonumber
		\end{align}
		where the component $K_1$ has been doubled in a parallel way in the first term on the RHS of the equality (see Figure \ref{fig:doubling}).
	\end{enumerate}
\end{prop}
Part (a) of Proposition \ref{prop:Erec} is based on the fact that any component in the trivial representation can be unlinked since the full-monodromy matrix $\cB$ acts as the identity in the case $j_1=0$, and also, on the factorization property \eqref{eq:Efact}. Part (b) follows from the spin decomposition rule $j \otimes 1/2 = (j+1/2) \oplus (j-1/2)$ and the fusion property \eqref{eq:Efus}.

\begin{figure}[]
	\begin{tikzpicture}[scale=1,line width=\lw, baseline={([yshift=-\eseq]current bounding box.center)}]
		\draw[decoration={markings, mark=at position 0.6 with {\arrow{>}}},postaction=decorate] (0:0.8cm) arc(0:-360:0.8cm and 1cm);
		\node at (0,-1.2cm) {\footnotesize $j_1$};
	\end{tikzpicture} 
	$ \quad = \quad $
	\begin{tikzpicture}[scale=1,line width=\lw, baseline={([yshift=-\eseq]current bounding box.center)}]
		\draw[decoration={markings, mark=at position 0.6 with {\arrow{>}}},postaction=decorate] (0:0.8cm) arc(0:-360:0.8cm and 1cm);
		\draw[color=white,fill=white] (0.12cm,1cm) circle(0.12cm);
		\draw[color=white,fill=white] (0.12cm,-1cm) circle(0.12cm);
		\begin{scope}[xshift=0.25cm]
			\draw[decoration={markings, mark=at position 0.6 with {\arrow{>}}},postaction=decorate] (0:0.8cm) arc(0:-360:0.8cm and 1cm);
		\end{scope}
		\node at (-0.5cm,-1.2cm) {\footnotesize $\frac{1}{2}$};
		\node at (1cm,-1.2cm) {\footnotesize $j_1-\frac{1}{2}$};
	\end{tikzpicture} 
	$ - \quad $
	\begin{tikzpicture}[scale=1,line width=\lw, baseline={([yshift=-\eseq]current bounding box.center)}]
		\draw[decoration={markings, mark=at position 0.6 with {\arrow{>}}},postaction=decorate] (0:0.8cm) arc(0:-360:0.8cm and 1cm);
		\node at (0cm,-1.2cm) {\footnotesize $j_1-1$};
	\end{tikzpicture} 
	\centering
	\caption{Diagrammatic representation of equation \eqref{eq:Erec}. The equality depicted above extends when the diagrams are part of a more complicated link.}
	\label{fig:doubling}
\end{figure}
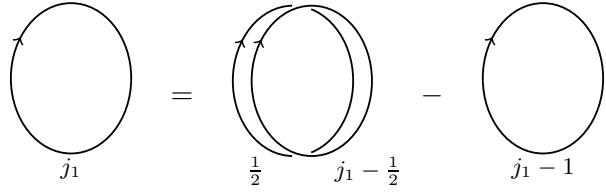

\subsection{Kauffman bracket}\label{ssec:KauffmanBracket}

The strategy for establishing the role of the Askey--Wilson algebra in the CS theory will rely on connecting the expectation values of the Wilson loops in the fundamental representation of $\su_2$ with the Kauffman bracket polynomial. It is relevant to point out a parallel with the relation between the AW algebra and the Kauffman skein algebra that was obtained in \cite{BP,Cooke,CFGPRV,Hik}.  

We start with a definition of Kauffman's bracket polynomial.
\begin{defi}\label{def:VB} \cite{Kau87}
	The bracket polynomial $V_B(L;x)$ for a non-oriented link $L$ is defined by
	\begin{enumerate}[label=(\roman*)]
		\item $V_B(L)=V_B(L')$ if $L$ and $L'$ are regular isotopic;
		\item $V_B\left(
		      \begin{tikzpicture}[line width=\lw,baseline={([yshift=-0.1cm] current bounding box.center)}]
		      	\draw[\colfund] circle(\lxk);
		      \end{tikzpicture} 
		      \right)=-(x^2+x^{-2})$;
		\item $V_B\left( 
		      \begin{tikzpicture}[line width=\lw,baseline={([yshift=-0.1cm] current bounding box.center)}]
			      \draw[\colfund] (-\lxk,-\lxk) -- (0.1cm,0.1cm) arc(135:-135:0.141cm) -- (\sk,-\sk) (-\sk,\sk) -- (-\lxk,\lxk);
			  \end{tikzpicture} 
		  	  \right) = -x^3 V_B\left( \ 
		  	  \begin{tikzpicture}[line width=\lw,baseline={([yshift=-0.1cm] current bounding box.center)}]
		  	  	\draw[\colfund] (0,-\lxk) to [bend right = \bd] (0,\lxk);
		  	  \end{tikzpicture} 
		  	  \ \right), \qquad 
		  	  V_B\left( 
		  	  \begin{tikzpicture}[line width=\lw,baseline={([yshift=-0.1cm] current bounding box.center)}]
		  	  	\draw[\colfund] (-\lxk,-\lxk) -- (-\sk,-\sk) (\sk,\sk) -- (0.1cm,0.1cm) arc(135:-135:0.141cm) -- (-\lxk,\lxk);
		  	  \end{tikzpicture}
		  	  \right) = -x^{-3} V_B\left( \ 
		  	  \begin{tikzpicture}[line width=\lw,baseline={([yshift=-0.1cm] current bounding box.center)}]
		  	  	\draw[\colfund] (0,-\lxk) to [bend right = \bd] (0,\lxk);
		  	  \end{tikzpicture} 
		  	  \ \right)$;
		\item $V_B\left(  
		      \begin{tikzpicture}[line width=\lw,baseline={([yshift=-0.1cm] current bounding box.center)}]
		      	\draw[\colfund]
		      	(-\lxk,-\lxk) -- (-\sk,-\sk) (\sk,\sk) -- (\lxk,\lxk);
		      	\draw[\colfund]
		      	(-\lxk,\lxk) -- (\lxk,-\lxk);
		      \end{tikzpicture} 
		      \right)=x V_B\left( 
		      \begin{tikzpicture}[line width=\lw,baseline={([yshift=-0.1cm] current bounding box.center)}]
		      	\draw[\colfund]
		      	(-\lxk,\lxk) to [bend right = \bd] (\lxk,\lxk);
		      	\draw[\colfund]
		      	(-\lxk,-\lxk) to [bend left = \bd] (\lxk,-\lxk);
		      \end{tikzpicture}
		      \right) + x^{-1} V_B\left( 
		      \begin{tikzpicture}[line width=\lw,baseline={([yshift=-0.1cm] current bounding box.center)}]
		      	\draw[\colfund]
		      	(-\lxk,-\lxk) to [bend right = \bd] (-\lxk,\lxk);
		      	\draw[\colfund]
		      	(\lxk,-\lxk) to [bend left = \bd] (\lxk,\lxk);
		      \end{tikzpicture}
		      \right)$.
	\end{enumerate}
	In the previous equations, the dependence on the variable $x$ is implicit, and it is understood in (iii) and (iv) that the equalities hold for link diagrams which look the same everywhere except at some small region where they differ as illustrated. 
\end{defi}
Note that given an oriented link $L$, one can compute its bracket polynomial $V_B(L)$ by removing the orientation. The next result gives precisely the connection between the CS expectation values and the bracket polynomial.
\begin{prop}\label{prop:EVB} \cite{Gua} Let $L$ be an oriented framed link with all components associated to the fundamental (spin $1/2$) representation of $\su_2$. Then the following relation holds
	\begin{equation}
		\E(L) = \exp\left( - \frac{i \pi}{2} w(L) \right) V_B(L;iq^{\frac{1}{2}}). \label{eq:EVB}
	\end{equation}
\end{prop}
The idea of the proof is to show that $\exp\left( \frac{i \pi}{2} w(L) \right)\E(L)$ satisfies the properties (i)--(iv) with $x=iq^{\frac{1}{2}}$ which uniquely define the bracket polynomial $V_B(L;iq^{\frac{1}{2}})$. 

\subsection{Askey--Wilson relations from the Wilson loop expectation values}\label{ssec:AWWL}
We are now ready to prove the following theorem.
\begin{thm}\label{thm:CSAW}
	In the Chern--Simons theory on $\bR^3$ with gauge group $SU(2)$ and vertical framing, the Wilson loop expectation values \eqref{eq:vevWLVF} of the diagrams \eqref{eq:Ci}--\eqref{eq:C123} satisfy the Askey--Wilson relations \eqref{eq:AW1}--\eqref{eq:AW3} under the correspondence \eqref{eq:corres}.
\end{thm} 
  
Let us consider the product $\A_{12}\A_{23}$. To simplify the drawings in what follows, we will represent the tangle diagrams by viewing them from the top. In this case, the three vertical strands look as three punctures on a plane, and the blue loops enclose them. Hence, in this diagrammatic representation, the product looks like
\begin{equation}
	\A_{12}\A_{23} = \begin{tikzpicture}[scale=1,line width=\lw,baseline={([yshift=-\eseq]current bounding box.center)}]
		\pppunct{0}{0}
		\circpp{0}{0}{decorate}
		\circpplb{\xdp}{0}{decorate}
	\end{tikzpicture} \ . \label{eq:C12C23}
\end{equation}
The three strands (represented by the punctures) are associated to three spins $j_1,j_2$ and $j_3$. One can apply equation \eqref{eq:Erec} recursively to reduce these three spins $j_1$, $j_2$ and $j_3$ to $0$ or $1/2$, at the cost of adding a finite number of parallel strands. Because of the property \eqref{eq:Ereptriv}, the strands with spin $0$ can be removed. Therefore, one can express the product \eqref{eq:C12C23} as a linear combination of diagrams with spin-$1/2$ (blue) components:  
\begin{equation}
	\E\left( 
	\begin{tikzpicture}[scale=1,line width=\lw,baseline={([yshift=-\eseq]current bounding box.center)}]
		\pppunct{0}{0}
		\circpp{0}{0}{decorate}
		\circpplb{\xdp}{0}{decorate}
	\end{tikzpicture} \right) = \sum_{n_1,n_2,n_3} \beta_{n_1,n_2,n_3}  \E\left( 
	\begin{tikzpicture}[scale=1,line width=\lw,baseline={([yshift=-\eseq]current bounding box.center)}]
		\punctni{0}{0.05cm}{1}
		\punctni{\xdp}{0.05cm}{2}
		\punctni{2*\xdp}{0.05cm}{3}
		\circpp{0}{0}{decorate}
		\circpplb{\xdp}{0}{decorate}
	\end{tikzpicture}
	\right) \ . \label{eq:lincombdiag}
\end{equation}
In the previous equation, we have used the symbol $\{{\color{\colfund} \bullet}\}$ to represent a set of punctures all associated to the spin $1/2$, and we have indicated the number of punctures in each set by some non-negative integers $n_i$. The coefficient $\beta_{n_1,n_2,n_3}$ can be computed from \eqref{eq:Erec}, but its exact form will not be relevant. In what follows, for simplicity the integers $n_i$ will not be indicated explicitly on the diagrams anymore.

It is now possible to compute the bracket polynomial $V_B(L;iq^\frac{1}{2})$ of the diagrams which appear on the RHS in \eqref{eq:lincombdiag} by removing the orientation. More specifically, one can use the property (iv) of the bracket polynomial to simplify the crossings of the diagrams. Such a computation can be found in \cite{CFGPRV}, but we reproduce it here:
\begin{align}
	&V_B\left( 
	\begin{tikzpicture}[scale=1,line width=\lw,baseline={([yshift=-\eseq]current bounding box.center)}]
		\pppunctn{0}{0}
		\circpp{0}{0}{}
		\circpplb{\xdp}{0}{}
	\end{tikzpicture} \ ; iq^{\frac{1}{2}} \right) \nonumber \\
	&=  iq^{\frac{1}{2}} V_B\left( 
	\begin{tikzpicture}[scale=1,line width=\lw,baseline={([yshift=-\eseq]current bounding box.center)}]
		\pppunctn{0}{0}
		\draw[\colfund] (\xdp,-\rp) arc (-90:90:{\rp} and {\rpb/2+\rp/2}) arc (90:270:{\rp} and {\rpb}) (\xdp+\rp-2*\sp,-\rp) -- (\xdp+\xdp,-\rp) arc (-90:90:\rp) -- (0,\rp) arc (90:270:\rp) -- (\xdp,-\rp);
	\end{tikzpicture} \ ; iq^{\frac{1}{2}} \right)
	-iq^{-\frac{1}{2}} V_B\left( 
	\begin{tikzpicture}[scale=1,line width=\lw,baseline={([yshift=-\eseq]current bounding box.center)}]
		\pppunctn{0}{0}
		\draw[\colfund] (\xdp-0.1cm,-0.3cm) .. controls +(-0.2cm,0.05cm) and +(0,-0.25cm) .. (\rp,0) arc (0:270:\rp) -- (\rp,-\rp) .. controls +(0.4cm,0) and +(0,-0.3cm) .. (2*\xdp-\rp,0) arc (180:-90:\rp) -- (\xdp+0.1cm,-\rp) ;
	\end{tikzpicture} \ ; iq^{\frac{1}{2}} \right) \\
	&= -q V_B\left( 
	\begin{tikzpicture}[scale=1,line width=\lw,baseline={([yshift=-\eseq]current bounding box.center)}]
		\pppunctn{0}{0}
		\circppu{0}{0}{}
	\end{tikzpicture} \ ; iq^{\frac{1}{2}} \right) 
	+ V_B\left( 
	\begin{tikzpicture}[scale=1,line width=\lw,baseline={([yshift=-\eseq]current bounding box.center)}]
		\pppunctn{0}{0}
		\circps{\xdp}{0}{}
		\circppp{0}{0}{}
	\end{tikzpicture} \ ; iq^{\frac{1}{2}} \right) \label{eq:simplCrossVB}   \\
	&\phantom{=} + V_B\left( 
	\begin{tikzpicture}[scale=1,line width=\lw,baseline={([yshift=-\eseq]current bounding box.center)}]
		\pppunctn{0}{0}
		\circp{0}{0}{}
		\circp{2*\xdp}{0}{}
	\end{tikzpicture} \ ; iq^{\frac{1}{2}} \right)
	- q^{-1} V_B\left( 
	\begin{tikzpicture}[scale=1,line width=\lw,baseline={([yshift=-\eseq]current bounding box.center)}]
		\pppunctn{0}{0}
		\circppd{0}{0}{}
	\end{tikzpicture} \ ; iq^{\frac{1}{2}} \right). \nonumber
\end{align}
Using the relation \eqref{eq:EVB}, one can write the previous equation in terms of the Wilson loops expectation values $\E(L)$ with some exponential phases. Suppose the diagram on the LHS of \eqref{eq:simplCrossVB} has writhe $w(L)$ and $n_2$ punctures in the second set $\{{\color{\colfund} \bullet}\}$. By inserting back the orientations (all anti-clockwise), we find that the writhe numbers of the diagrams in \eqref{eq:simplCrossVB} are either $w(L)$ or $w(L)-4n_2$, depending whether there are two or zero loops which enclose the second set of punctures. Since $n_2$ is an integer, all the exponential phases of the diagrams reduce to the factor $\exp\left(i \pi w(L)/2 \right)$, which can be simplified from the equation. Hence
\begin{align}
	\E\left( 
	\begin{tikzpicture}[scale=1,line width=\lw,baseline={([yshift=-\eseq]current bounding box.center)}]
		\pppunctn{0}{0}
		\circpp{0}{0}{decorate}
		\circpplb{\xdp}{0}{decorate}
	\end{tikzpicture} \right) 
	&= -q  \E\left( 
	\begin{tikzpicture}[scale=1,line width=\lw,baseline={([yshift=-\eseq]current bounding box.center)}]
		\pppunctn{0}{0}
		\circppu{0}{0}{decorate}
	\end{tikzpicture} \right) 
	+ \E\left( 
	\begin{tikzpicture}[scale=1,line width=\lw,baseline={([yshift=-\eseq]current bounding box.center)}]
		\pppunctn{0}{0}
		\circps{\xdp}{0}{decorate}
		\circppp{0}{0}{decorate}
	\end{tikzpicture} \right) \label{eq:EC12C23} \\
	&\phantom{=} + \E\left( 
	\begin{tikzpicture}[scale=1,line width=\lw,baseline={([yshift=-\eseq]current bounding box.center)}]
		\pppunctn{0}{0}
		\circp{0}{0}{decorate}
		\circp{2*\xdp}{0}{decorate}
	\end{tikzpicture} \right)
	- q^{-1} \E\left( 
	\begin{tikzpicture}[scale=1,line width=\lw,baseline={([yshift=-\eseq]current bounding box.center)}]
		\pppunctn{0}{0}
		\circppd{0}{0}{decorate}
	\end{tikzpicture} \right). \nonumber
\end{align}
We can proceed similarly for the product $\A_{23}\A_{12}$ and find
\begin{align}
	\E\left( 
	\begin{tikzpicture}[scale=1,line width=\lw,baseline={([yshift=-\eseq]current bounding box.center)}]
		\pppunctn{0}{0}
		\circpprb{0}{0}{decorate}
		\circpp{\xdp}{0}{decorate}
	\end{tikzpicture} \right)
	&= -q^{-1}  \E\left( 
	\begin{tikzpicture}[scale=1,line width=\lw,baseline={([yshift=-\eseq]current bounding box.center)}]
		\pppunctn{0}{0}
		\circppu{0}{0}{decorate}
	\end{tikzpicture}  \right) 
	+ \E\left( 
	\begin{tikzpicture}[scale=1,line width=\lw,baseline={([yshift=-\eseq]current bounding box.center)}]
		\pppunctn{0}{0}
		\circps{\xdp}{0}{decorate}
		\circppp{0}{0}{decorate}
	\end{tikzpicture}  \right) \label{eq:EC23C12}  \\
	&\phantom{=} + \E\left( 
	\begin{tikzpicture}[scale=1,line width=\lw,baseline={([yshift=-\eseq]current bounding box.center)}]
		\pppunctn{0}{0}
		\circp{0}{0}{decorate}
		\circp{2*\xdp}{0}{decorate}
	\end{tikzpicture} \right)
	- q \E\left( 
	\begin{tikzpicture}[scale=1,line width=\lw,baseline={([yshift=-\eseq]current bounding box.center)}]
		\pppunctn{0}{0}
		\circppd{0}{0}{decorate}
	\end{tikzpicture} \right). \nonumber
\end{align}
Using an argument similar to the one which lead to \eqref{eq:lincombdiag}, each set of punctures $\{{\color{\colfund} \bullet}\}$ can be put back to a single puncture $\bullet$ with spin $j_i$ for $i=1,2,3$ in equations \eqref{eq:EC12C23} and \eqref{eq:EC23C12}. Therefore, these two equations imply
\begin{align}
	&q\E\left( 
	\begin{tikzpicture}[scale=1,line width=\lw,baseline={([yshift=-\eseq]current bounding box.center)}]
		\pppunct{0}{0}
		\circpp{0}{0}{decorate}
		\circpplb{\xdp}{0}{decorate}
	\end{tikzpicture}\right)-q^{-1}\E\left( 
	\begin{tikzpicture}[scale=1,line width=\lw,baseline={([yshift=-\eseq]current bounding box.center)}]
		\pppunct{0}{0}
		\circpprb{0}{0}{decorate}
		\circpp{\xdp}{0}{decorate}
	\end{tikzpicture} \right) +(q^2-q^{-2})  \E\left( 
	\begin{tikzpicture}[scale=1,line width=\lw,baseline={([yshift=-\eseq]current bounding box.center)}]
		\pppunct{0}{0}
		\circppu{0}{0}{decorate}
	\end{tikzpicture} \right) \nonumber  \\
	&= 
	(q-\qi)\left\{ \E\left( 
	\begin{tikzpicture}[scale=1,line width=\lw,baseline={([yshift=-\eseq]current bounding box.center)}]
		\pppunct{0}{0}
		\circps{\xdp}{0}{decorate}
		\circppp{0}{0}{decorate}
	\end{tikzpicture}  \right)  + \E\left( 
	\begin{tikzpicture}[scale=1,line width=\lw,baseline={([yshift=-\eseq]current bounding box.center)}]
		\pppunct{0}{0}
		\circp{0}{0}{decorate}
		\circp{2*\xdp}{0}{decorate}
	\end{tikzpicture} \right)\right\}, \label{eq:EAW1}
\end{align}
which is the defining relation \eqref{eq:AW1} of the special Askey--Wilson algebra. The relations \eqref{eq:AW2} and \eqref{eq:AW3} can be obtained by conjugating the diagrams in \eqref{eq:EAW1} by the braids $\sigma_1\sigma_2$ and $(\sigma_1\sigma_2)^{-1}$ respectively. With a similar method as for \eqref{eq:AW1}, relation \eqref{eq:AW4} can be shown to hold at the level of the bracket polynomial $V_B(L)$ \cite{CFGPRV}. Observing that each index $1,2,3$ of the generators of the special AW algebra is repeated twice in each term of this relation, one deduces again that the exponential factors which appear in \eqref{eq:EVB} all simplify. Therefore, relation \eqref{eq:AW4} also holds for the CS link invariant.    

\subsection{Connection with the Temperley--Lieb algebra}\label{ssec:TL}
It is known that Kauffman's bracket polynomial is connected to the Jones polynomial and to the Temperley--Lieb algebra \cite{Jo85,Kau87,TL}. More recently, it has been shown through the study of a generalization of the Schur--Weyl duality for $U_q(\gsl_2)$ that the Temperley--Lieb algebra is isomorphic to a quotient of the Askey--Wilson algebra \cite{CVZ2}. We offer now an interpretation of this isomorphism in terms of diagrams.

\begin{defi}\cite{TL} 
	The Temperley--Lieb algebra $TL_3(q)$ is generated by $e_1$ and $e_2$ with the following defining relations
	\begin{align}
		&e_1^2=(q+\qi)e_1, \quad e_2^2=(q+\qi)e_2, \label{eq:TL1} \\
		&e_1e_2e_1=e_1, \quad e_2e_1e_2=e_2. \label{eq:TL2} 
	\end{align}
\end{defi}
To establish precisely the connection between the bracket polynomial $V_B(L;x)$ and the Temperley--Lieb algebra $TL_3(q)$, one defines the ``hook'' diagrams (see \cite{Kau87,Kau90})
\begin{equation}
	E_1 =
	\begin{tikzpicture}[scale=1,line width=\lw,baseline={([yshift=-\eseq]current bounding box.center)}]
		\draw[\colfund] (0,0) arc(180:0:\xd/2);
		\draw[\colfund] (0,\yd) arc(180:360:\xd/2);
		\draw[\colfund] \strand{2*\xd}{0};
	\end{tikzpicture} \ , \qquad
	E_2 =
	\begin{tikzpicture}[scale=1,line width=\lw,baseline={([yshift=-\eseq]current bounding box.center)}]
		\draw[\colfund] \strand{0}{0};
		\draw[\colfund] (\xd,0) arc(180:0:\xd/2);
		\draw[\colfund] (\xd,\yd) arc(180:360:\xd/2);
	\end{tikzpicture} \ .
\end{equation}
From the properties (iii) and (iv) of Definition \ref{def:VB}, it is seen that the following equalities hold for the value $V_B(L;x)$ of the illustrated diagrams:  
\begin{equation}
	\begin{tikzpicture}[line width=\lw,baseline={([yshift=-0.1cm] current bounding box.center)}]
		\draw[\colfund] {[rounded corners = 5pt] \ucrossl{0}{0}} (\xd,\yd) arc(0:180:\xd/2) {[rounded corners = 5pt] \ucrossr{\xd}{0}};
		\draw[\colfund] (0,\yd+\yd) arc(180:360:\xd/2);
	\end{tikzpicture} = -x^{3} \
	\begin{tikzpicture}[scale=1,line width=\lw,baseline={([yshift=-\eseq]current bounding box.center)}]
		\draw[\colfund] (0,0) arc(180:0:\xd/2);
		\draw[\colfund] (0,\yd+\yd) arc(180:360:\xd/2);
	\end{tikzpicture} =  x \ 
	\begin{tikzpicture}[scale=1,line width=\lw,baseline={([yshift=-\eseq]current bounding box.center)}]
		\draw[\colfund] (0,0) arc(180:0:\xd/2);
		\draw[\colfund] (0,\yd) arc(180:360:\xd/2);
		\draw[\colfund] (0,\yd) arc(180:0:\xd/2);
		\draw[\colfund] (0,\yd+\yd) arc(180:360:\xd/2);
	\end{tikzpicture} \ + \ x^{-1} \
	\begin{tikzpicture}[scale=1,line width=\lw,baseline={([yshift=-\eseq]current bounding box.center)}]
		\draw[\colfund] (0,0) -- (0,\yd) arc(180:0:\xd/2) -- (\xd,0);
		\draw[\colfund] (0,\yd+\yd) arc(180:360:\xd/2);
	\end{tikzpicture}  \ . \label{eq:VBTL1} 
\end{equation}
Moreover, the regular isotopy invariance property (i) of Definition \ref{def:VB} implies for the value of $V_B(L;x)$
\begin{equation}
	\begin{tikzpicture}[line width=\lw,baseline={([yshift=-0.1cm] current bounding box.center)}]
		\draw[\colfund] (0,0) arc(180:0:\xd/2);
		\draw[\colfund] \strand{2*\xd}{0} arc(0:180:\xd/2) arc(0:-180:\xd/2) -- ++\strand{0}{0} arc(180:0:\xd/2) arc(180:360:\xd/2) -- ++\strand{0}{0};
		\draw[\colfund] (0,\yd+\yd+\yd) arc(180:360:\xd/2);
	\end{tikzpicture} \ = \
	\begin{tikzpicture}[line width=\lw,baseline={([yshift=-0.1cm] current bounding box.center)}]
		\draw[\colfund] (0,0) arc(180:0:\xd/2);
		\draw[\colfund] (0,\yd+\yd+\yd) arc(180:360:\xd/2);
		\draw[\colfund] \strand{2*\xd}{0}  -- ++\strand{0}{0} -- ++\strand{0}{0};
	\end{tikzpicture} \ , \qquad
	\begin{tikzpicture}[line width=\lw,baseline={([yshift=-0.1cm] current bounding box.center)}]
		\draw[\colfund] \strand{0}{0} arc(180:0:\xd/2) arc(-180:0:\xd/2) -- ++\strand{0}{0} arc(0:180:\xd/2) arc(360:180:\xd/2) -- ++\strand{0}{0};
		\draw[\colfund] (\xd,0) arc(180:0:\xd/2);
		\draw[\colfund] (\xd,\yd+\yd+\yd) arc(180:360:\xd/2);
	\end{tikzpicture} \ = \
	\begin{tikzpicture}[line width=\lw,baseline={([yshift=-0.1cm] current bounding box.center)}]
		\draw[\colfund] \strand{0}{0}  -- ++\strand{0}{0} -- ++\strand{0}{0};
		\draw[\colfund] (\xd,0) arc(180:0:\xd/2);
		\draw[\colfund] (\xd,\yd+\yd+\yd) arc(180:360:\xd/2);
	\end{tikzpicture} \ . \label{eq:VBTL2}
\end{equation}
Therefore, one deduces from equations \eqref{eq:VBTL1} and \eqref{eq:VBTL2} that the bracket polynomial $V_B(L;x=iq^{\frac{1}{2}})$ of the diagrams $E_1,E_2$ satisfy the defining relations \eqref{eq:TL1} and \eqref{eq:TL2} of $TL_3(q)$, respectively with $e_i \mapsto E_i$, for $i=1,2$.
 
In \cite{CVZ2}, the centralizer of the diagonal action of the algebra $U_q(\gsl_2)$ in the tensor product of three finite irreducible representations of $U_q(\gsl_2)$ of spins $j_1,j_2,j_3$ is conjectured to be isomorphic to a quotient of the special Askey--Wilson algebra $\saw(3)$. The conjecture is proven for $j_1=j_2=j_3=1/2$, in which case the centralizer is known to be isomorphic to the Temperley--Lieb algebra $TL_3(q)$. As a result, an explicit isomorphism between a quotient of $\saw(3)$ and $TL_3(q)$ was obtained in \cite{CVZ2} (see therein Theorem 6.2). In the notations of the present paper, this isomorphism maps the (quotiented) generators of $\saw(3)$ to elements of $TL_3(q)$ as follows
\begin{align}
	\C_{12} \mapsto (q^3+q^{-3}) - (q-\qi)^2e_1, \label{eq:isoAWTL1} \\
	\C_{23} \mapsto (q^3+q^{-3}) - (q-\qi)^2e_2. \label{eq:isoAWTL2}
\end{align}    

We have shown that, in the CS theory, the tangle diagrams $\A_{12}$ and $\A_{23}$ defined in \eqref{eq:Cij} are in correspondence with the generators $\C_{12}$ and $\C_{23}$ of the AW algebra. In the case where the three vertical strands of these tangle diagrams have spins $j_1=j_2=j_3=1/2$, we can remove the orientations and compute the bracket polynomial, which is related to the CS expectation value as stated in Proposition \ref{prop:EVB}. It is straightforward to use the properties of $V_B(L;iq^{\frac{1}{2}})$ to find
\begin{equation}
	V_B\left(
	\begin{tikzpicture}[scale=1,line width=\lw,baseline={([yshift=-\eseq]current bounding box.center)}]
		\draw[\colfund] \cstrand{0}{0};
		\draw[\colfund] \cstrand{\xd}{0};
		\circss{0}{\ly+0.05cm}{}
	\end{tikzpicture} \ ; iq^{\frac{1}{2}}
	\right) =
	(q^3+q^{-3}) V_B\left( \ 
	\begin{tikzpicture}[scale=1,line width=\lw,baseline={([yshift=-\eseq]current bounding box.center)}]
		\draw[\colfund] \strand{0}{0};
		\draw[\colfund] \strand{\xd}{0};
	\end{tikzpicture} \ ; iq^{\frac{1}{2}}
	\right) -(q-\qi)^2 V_B\left( \ 
	\begin{tikzpicture}[scale=1,line width=\lw,baseline={([yshift=-\eseq]current bounding box.center)}]
		\draw[\colfund] (0,0) arc(180:0:\xd/2);
		\draw[\colfund] (0,\yd) arc(180:360:\xd/2);
	\end{tikzpicture} \ ; iq^{\frac{1}{2}} \right). \label{eq:VBTLAW}
\end{equation}
By adding a third vertical strand on the right (resp. left) of the diagrams which appear in \eqref{eq:VBTLAW}, we recognize the tangle diagram $\A_{12}$ (resp. $\A_{23}$) on the LHS and the Temperley--Lieb hook diagram $E_1$ (resp. $E_2$) in the second term of the RHS. Therefore, comparing with \eqref{eq:isoAWTL1} and \eqref{eq:isoAWTL2}, equation \eqref{eq:VBTLAW} can be seen as a diagrammatic representation of the isomorphism between the quotient of the AW algebra and the TL algebra established in \cite{CVZ2}.
	
\section{Askey--Wilson algebra in the Reshetikhin--Turaev link invariant construction}\label{sec:AWRT}       
In this section, we will consider the quantum algebra $\Usl$, a $q$-deformation of the universal enveloping algebra of $\su_2$, and we will recall how its universal $R$-matrix can be used to construct a link invariant. We will then show that the tangle diagrams \eqref{eq:Ci}--\eqref{eq:C123} also lead to the Askey--Wilson algebra relations in this framework. It will be assumed throughout this section that $q$ is a generic complex number.      

\subsection{$U_q(\mathfrak{su}_2)$ algebra} \label{ssec:Uqsu2}

The quantum algebra $\Usl$ is the associative algebra generated by $E$, $F$ and $q^{H}$ with the defining relations
\begin{equation}
	q^{H}E=q  Eq^H, \quad q^{H}F=q^{-1} Fq^H, \quad [E,F]=[2H]_q, \label{eq:Uqsl2}
\end{equation}
where $[X,Y]=XY-YX$. For future convenience, we define
\begin{equation}
	\mu:=q^{2H}. \label{eq:mu}
\end{equation}
The following Casimir element generates the center of $\Usl$:
\begin{equation}
	\Q=(q-\qi)^2 FE+q^{2H+1}+ q^{-2H-1}. \label{eq:Casimir}
\end{equation} 

We now describe some relevant features of the quasitriangular Hopf algebra structure of $\Usl$. The comultiplication, or coproduct, is an algebra homomorphism $\Delta:\Usl \rightarrow \Usl \otimes \Usl $ which is defined on the generators by 
\begin{equation}
	\Delta(E)=E \otimes q^{-H} +q^H \otimes E , \quad  \Delta(F)=F \otimes q^{-H} +q^{H} \otimes F, \quad \Delta(q^H)=q^H\otimes q^H, \label{eq:comult}
\end{equation}
and which is coassociative 
\begin{equation}
	(\Delta \otimes \id)\circ\Delta= (\id \otimes \Delta)\circ\Delta. \label{eq:coasso}
\end{equation}
The antipode is an algebra anti-automorphism $S:\Usl \to \Usl$ defined by
\begin{equation}
	S(E)=-\qi E, \quad  S(F)=-qF, \quad S(q^H)=q^{-H}. \label{eq:antipode}
\end{equation}
The universal $R$-matrix of $\Usl$ is an invertible element $\cR\in \Usl \otimes \Usl$ which satisfies
\begin{equation}
	\Delta^{op}(x)=\cR \Delta(x) \cR^{-1} \quad \forall x\in \Usl. \label{eq:RD}
\end{equation}
In the previous equation, $\Delta^{op}$ is the opposite comultiplication and is defined by $\Delta^{op}=\tau \circ \Delta$, where $\tau(x \otimes y)=y \otimes x$ for $x,y\in \Usl$. The universal $R$-matrix also satisfies
\begin{equation}
	(\id \otimes \Delta)(\cR)=\cR_{13}\cR_{12}, \qquad (\Delta \otimes \id)(\cR)=\cR_{13}\cR_{23}, \label{eq:idDR}
\end{equation}
and the Yang--Baxter equation
\begin{equation}
	\cR_{12}\cR_{13}\cR_{23}=\cR_{23}\cR_{13}\cR_{12}. \label{eq:YBE}
\end{equation}
We have used standard notations: if $\cR=\cR^\alpha\otimes \cR_\alpha$, then $\cR_{12}=\cR \otimes 1 = \cR^\alpha\otimes \cR_\alpha \otimes 1$, $\cR_{23}=1 \otimes \cR = 1\otimes \cR^\alpha\otimes \cR_\alpha$ and $\cR_{13}=(\id\otimes \tau)(\cR_{12})=\cR^\alpha\otimes1\otimes  \cR_\alpha$ (the sum w.r.t. $\alpha$ is understood). We will also denote $\cR_{21}=\tau(\cR)=\cR_\alpha\otimes \cR^\alpha$. The universal $R$-matrix of $\Usl$ is explicitly given by \cite{Dr}
\begin{equation}
	\cR= \sum_{k=0}^\infty \frac{(q-q^{-1})^k}{[k]_q!} q^{-k(k+1)/2} (  F \otimes E )^k  (q^{kH}\otimes q^{-kH}) q ^{2(H\otimes H)}, \label{eq:uR}
\end{equation}
where $[n]_q!:=[n]_q[n-1]_q\dots [2]_q[1]_q$ and, by convention, $[0]_q!:=1$. The inverse is given by
\begin{equation}
	\cR^{-1}=(S \otimes \id)(\cR)=(\id \otimes S^{-1})(\cR). \label{eq:uRi}
\end{equation}

\subsection{Representations of $U_q(\mathfrak{su}_2)$} \label{ssec:RepU}

Like the Lie algebra $\su_2$, the algebra $\Usl$ has an irreducible spin-$j$ representation $V_j$ of finite dimension $2j+1$ for each $j=0,\frac{1}{2},1,...$ The basis vectors of $V_j$ will be denoted by $\ket{j,m}$, and the dual basis vectors by $\bra{j,m}$, for $m=-j,-j+1,...,j$. We have the following orthogonality and completeness relations:
\begin{equation}
	\braket{j,m|j,n} = \delta_{mn}, \qquad \sum_{m=-j}^j \ket{j,m}\bra{j,m} = \id. \label{eq:orthcompl}
\end{equation}
The representation of $\Usl$ on the basis vectors $\ket{j,m}$, for $m=-j,...,j$, is given by
\begin{align}
	&E\ket{j,m} = [j-m]_q \ket{j,m+1}, \label{eq:actE} \\
	&F\ket{j,m} = [j+m]_q \ket{j,m-1}, \label{eq:actF} \\
	&H\ket{j,m} = m \ket{j,m}. \label{eq:actH}
\end{align}
One deduces from \eqref{eq:actE}--\eqref{eq:actH} the constant action of the Casimir element \eqref{eq:Casimir} on $V_j$, for $m=-j,...,j$:  
\begin{equation}
	\Q\ket{j,m} = \chi_j \ket{j,m}, \qquad \text{where } \chi_j:=q^{2j+1}+q^{-2j-1}.
\end{equation}
As for the Lie algebra $\su_2$, the tensor product decomposition rule for two irreducible representations of spins $j_1$ and $j_2$ of $\Usl$ is $V_{j_1} \otimes V_{j_2} = \bigoplus_{j=|j_1-j_2|}^{j_1+j_2} V_j$. We will denote the tensor product of two basis vectors as $\ket{j_1,m_1} \otimes \ket{j_2,m_2} = \ket{j_1,m_1;j_2,m_2}$.

The universal $R$-matrix \eqref{eq:uR} with the first space in the fundamental (spin-$1/2$) representation of $\Usl$ is given by the following two by two matrix:
\begin{equation}\label{eq:Lm}
 L^-_{12}=\begin{pmatrix}
 q^H & 0\\
 \frac{q-q^{-1}}{q^{1/2}} E & q^{-H}
     \end{pmatrix}.
\end{equation}
Similarly, the element $\cR_{21}$ with the first space in the spin-$1/2$ representation is given by
\begin{equation}\label{eq:Lp}
 L^+_{12}=\begin{pmatrix}
 q^H &  \frac{q-q^{-1}}{q^{1/2}} F\\
 0 & q^{-H}
     \end{pmatrix}.
\end{equation}
Then, one can show that the universal Yang--Baxter equation \eqref{eq:YBE} with two spaces in the spin-$1/2$ representation leads to the FRT relations \cite{FRT}
\begin{eqnarray}
 &&R_{12} L^-_{13}  L^-_{23} = L^-_{23}  L^-_{13} R_{12}, \label{eq:FRT1} \\
 &&R_{12} L^+_{23}  L^+_{13} = L^+_{13}  L^+_{23} R_{12}, \label{eq:FRT2} \\
 && L^-_{13} R_{12} L^+_{23} = L^+_{23} R_{12} L^-_{13}, \label{eq:FRT3} 
\end{eqnarray}
where $R$, called the $R$-matrix, is the spin-$1/2$ representation of the universal $R$-matrix in both spaces:
\begin{equation}
 R=q^{1/2}\begin{pmatrix}
    1 & 0 & 0& 0\\
    0 & q^{-1} & 0& 0\\
    0 &1-q^{-2} & q^{-1} & 0\\
    0 &0 &0 & 1
   \end{pmatrix}.
\end{equation}
The three relations \eqref{eq:FRT1}--\eqref{eq:FRT3} are equivalent to the defining relations \eqref{eq:Uqsl2} of $\Usl$. 

We obtain the following relations by representing the Yang--Baxter equation \eqref{eq:YBE}  in one space
\begin{eqnarray}
	&&\cR_{23} L^-_{13}  L^-_{12} = L^-_{12}  L^-_{13} \cR_{23},\label{eq:RLL4} \\
	&&\cR_{12} L^+_{31}  L^+_{32} = L^+_{32}  L^+_{31} \cR_{12}, \label{eq:RLL5} \\
	&& L^+_{21} \cR_{13} L^-_{23} = L^-_{23} \cR_{13} L^+_{21}.\label{eq:RLL6} 
\end{eqnarray}

The coproduct can also be defined in this formalism. Indeed, one gets
\begin{equation}
 \Delta(L^+_{12})=L^+_{12}L^+_{13}, \quad \Delta(L^-_{12})=L^-_{13}L^-_{12}, \label{eq:idDL}
\end{equation}
where it is understood that $\Delta$ acts on the space which is not represented.

We will also denote by $M$ the spin-$1/2$ matrix representation of the element $\mu$ defined in \eqref{eq:mu}, given by
\begin{equation}
	M = \begin{pmatrix}
		q & 0 \\
		0 & \qi
	\end{pmatrix}. \label{eq:repmu}
\end{equation}

\subsection{Trace and link invariants}\label{ssec:TrLinkInv}

Let $n$ be a positive integer, and denote the tensor product of $\Usl$-representations of spins $j_1,...,j_n$ by
\begin{equation}
	V = V_{j_1} \otimes V_{j_2} \otimes ... \otimes V_{j_n}.
\end{equation}
This vector space has dimension $(2j_1+1)\times ... \times (2j_n+1)$ and basis vectors $\ket{j_1,m_1;...;j_n,m_n}$, where $m_i=-j_1,...,j_i$ for all $i=1,...,n$. We will now provide an action of the braid group $B_n$ on $V$ and define a trace which is a link invariant. The construction follows \cite{Resh,Tur}.

For $i=1,2,...,n-1$, we define the permutation operator $\Pi_{i,i+1}$ which exchanges the vectors in the positions $i$ and $i+1$ of the $n$-fold tensor product space $V$, that is  
\begin{equation}
	\Pi_{i,i+1} \ket{...;j_i,m_i;j_{i+1},m_{i+1};...} = \ket{...;j_{i+1},m_{i+1};j_i,m_i;...}. \label{eq:actperm}
\end{equation}
Using this permutation operator, we can define for $i=1,2,...,n-1$ the following universal braided $R$-matrix:
\begin{equation}
	\check \cR_{i,i+1} = \Pi_{i,i+1} \cR_{i,i+1}. \label{eq:cR}
\end{equation}
The action of the operator \eqref{eq:cR} on $V$ is obtained from the explicit expression \eqref{eq:uR} of the universal $R$-matrix, from the actions \eqref{eq:actE}--\eqref{eq:actH} of $\Usl$ and from the action \eqref{eq:actperm} of the permutation operator. The operator \eqref{eq:cR} has an inverse on $V$ which is given by
\begin{equation}
	\check \cR_{i,i+1}^{-1} = \cR_{i,i+1}^{-1}\Pi_{i,i+1}. \label{eq:cRi}
\end{equation}
Since $\check \cR_{i,i+1}$ acts only on the factors $i$ and $i+1$ of the $n$-fold tensor product space $V$, it is seen that the following equality of operators acting on $V$ holds
\begin{equation}
	\check \cR_{i,i+1} \check \cR_{j,j+1} = \check \cR_{j,j+1} \check \cR_{i,i+1} \quad \text{if } |i-j|>1.
\end{equation}
Moreover, using the Yang--Baxter equation \eqref{eq:YBE} and the action \eqref{eq:actperm} of the permutation operator, one can show that the following braid relation holds on $V$
\begin{equation}
	\check \cR_{i,i+1} \check \cR_{i+1,i+2} \check \cR_{i,i+1} = \check \cR_{i+1,i+2} \check \cR_{i,i+1} \check \cR_{i+1,i+2}.
\end{equation}
Therefore, the following defines an action of the braid group $B_n$ on $V$:
\begin{equation}
	\sigma_i^{\pm1} \ket{j_1,m_1;...;j_n,m_n} = \check\cR_{i,i+1}^{\pm 1} \ket{j_1,m_1;...;j_n,m_n}. \label{eq:actbraid}
\end{equation}
From \eqref{eq:actbraid}, it is seen that we can use interchangeably the braid group generators $\sigma_i$ and their corresponding universal braided $R$-matrices $\check\cR_{i,i+1}$ when acting on $V$. 

If $X$ is an operator which admits an action on $V$, we define its trace on $V$ by
\begin{equation}
	\Tr^{(j_1,...,j_n)}(X) := \sum_{m_1=-j_1}^{j_1}... \sum_{m_n=-j_n}^{j_n} \bra{j_1,m_1;...;j_n,m_n} X \ket{j_1,m_1;...;j_n,m_n}. 
\end{equation}
We also define the partial trace of $X$ on the first tensor factor of $V$, denoted $\Tr_1^{(j_1)}(X)$, as the operator acting on $V_{j_2} \otimes V_{j_3} \otimes ... \otimes V_{j_n}$ which satisfies for all $m_i,\ell_i=-j_i,...,j_i$, with $i=2,...,n$, the following relation 
\begin{align}
	&\bra{j_2,m_2;...;j_n,m_n}\Tr_1^{(j_1)}(X) \ket{j_2,\ell_2;...;j_n,\ell_n} \nonumber \\
	= &\sum_{m_1=-j_1}^{j_1} \bra{j_1,m_1;j_2,m_2...;j_n,m_n} X \ket{j_1,m_1;j_2,\ell_2;...;j_n,\ell_n}.
\end{align}
The partial trace of $X$ on the $n^\text{th}$ tensor factor of $V$ will be understood similarly and denoted $\Tr_n^{(j_n)}(X)$. Let us recall the following proposition which can be found in \cite{Resh,Tur}.
\begin{prop}\label{prop:TrEnh}
	The following equality of operators acting on the tensor product of any two spin representations of $\Usl$ holds 
	\begin{equation}
		\check \cR (\mu \otimes \mu) = (\mu \otimes \mu)  \check \cR. \label{eq:enhRcommu}
	\end{equation}
	 Moreover, we have the following equalities of operators acting on any spin-$j$ representation of $\Usl$
	\begin{align}
		&\Tr_1^{(j)}(\check \cR^{\pm1}(\mu \otimes 1))=q^{\pm 2 j(j+1)} \id, \label{eq:Trleft}\\
		&\Tr_2^{(j)}(\check \cR^{\pm1}(1 \otimes \mu^{-1}))=q^{\pm 2 j(j+1)} \id. \label{eq:Trright}
	\end{align}
\end{prop}
\begin{proof}
	See Appendix \ref{app:proofPropTrEnh}.
\end{proof}
\begin{rem}
	The element $\mu^{-1}$ with the commutativity property \eqref{eq:enhRcommu} and the partial trace property \eqref{eq:Trright} is called an ``enhancement'' in \cite{Tur}. The $R$-matrix together with such an enhancement is said to be an ``enhanced" Yang--Baxter operator.   
\end{rem}

Let $L$ be a link with components each associated to a spin. This link $L$ can be obtained as the closure of a braid $\sigma(L) \in B_n$ with $n$ strands labeled by a set of spins $\{j_i\}_{i=0}^n$, as illustrated in Figure \ref{fig:closurebraid}.
\begin{figure}[h]
	\begin{tikzpicture}[scale=1,line width=\lw, decoration={markings, mark=at position 0.46 with {\arrow{>}}}, baseline={([yshift=-\eseq]current bounding box.center)}]
		\draw[->] (0,0) node{\footnotesize $j_1$} ++(0,0.2cm) --  ++(0,0.8cm)  ++(0,0.7cm) -- ++(0,0.8cm);
		\node at (0,2.7cm) {\footnotesize $j_1$};
		\draw[->] (\xd,0) node{\footnotesize $j_2$} ++(0,0.2cm) -- ++(0,0.8cm)  ++(0,0.7cm) -- ++(0,0.8cm);
		\node at (\xd,2.7cm) {\footnotesize $j_2$};
		\node at (\xd+0.5cm,0.3cm) {...};
		\node at (\xd+0.5cm,2.2cm) {...};
		\draw[->] (2*\xd+0.3cm,0) node{\footnotesize $j_n$} ++(0,0.2cm) -- ++(0,0.8cm) ++(0,0.7cm) -- ++(0,0.8cm);
		\node at (2*\xd+0.3cm,2.7cm) {\footnotesize $j_n$};
		\draw (-0.2cm,1cm) rectangle +(2*\xd+0.7cm,0.7cm);
		\node at (\xd+0.2cm,1.35cm) {$\sigma(L)$};
	\end{tikzpicture} 
	$ \quad \longrightarrow \quad $  
	\begin{tikzpicture}[scale=1,line width=\lw, baseline={([yshift=-\eseq]current bounding box.center)}]
		\draw[decoration={markings, mark=at position 0.28 with {\arrow{>}}},postaction={decorate}] (0,0) --  ++(0,0.8cm)  ++(0,0.7cm) -- ++(0,0.8cm) arc(0:180:\ra) -- (-2*\ra,0) arc(180:360:\ra) -- (0,0);
		\draw[decoration={markings, mark=at position 0.174 with {\arrow{>}}},postaction={decorate}] (\xd,0) -- ++(0,0.8cm)  ++(0,0.7cm) -- ++(0,0.8cm) arc(0:180:\ra+\xd/2+0.2cm) -- (-2*\ra-0.4cm,0) arc(180:360:\ra+\xd/2+0.2cm) -- (\xd,0cm);
		\node at (\xd+0.5cm,0.3cm) {...};
		\node at (\xd+0.5cm,2.2cm) {...};
		\draw[decoration={markings, mark=at position 0.117 with {\arrow{>}}},postaction={decorate}] (2*\xd+0.3cm,0) -- ++(0,0.8cm) ++(0,0.7cm) -- ++(0,0.8cm) arc(0:180:\ra+\xd+0.55cm) -- (-2*\ra-0.8cm,0) arc(180:360:\ra+\xd+0.55cm) -- (2*\xd+0.3cm,0cm);
		\draw (-0.2cm,0.8cm) rectangle +(2*\xd+0.7cm,0.7cm);
		\node at (\xd+0.2cm,1.15cm) {$\sigma(L)$};
	\end{tikzpicture}
	\centering
	\caption{Closure of a braid $\sigma(L)$ representing a colored link $L$.}
	\label{fig:closurebraid}
\end{figure}
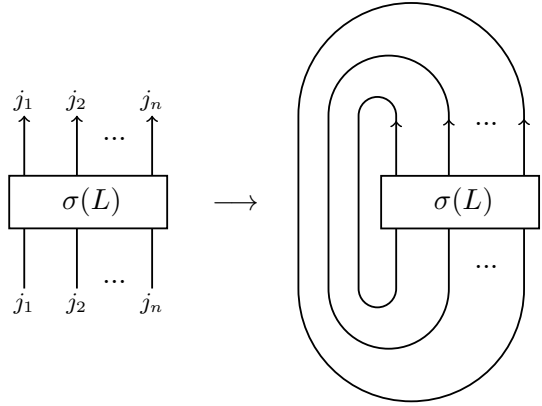
Note that the spins of the top endpoints of the braid must be equal to those of the bottom endpoints in order for the closure to make sense. In other words, the braid $\sigma(L)$ viewed as a permutation must leave the ordered set $\{j_i\}_{i=0}^n$ invariant. 

We can assign to the link $L$ the value of the ``quantum'' trace of the braid $\sigma(L)$ which represents it, defined as
\begin{equation}
	\Y(L):=\Tr^{(j_1,...,j_n)}(\sigma(L)\mu^{\otimes n}). \label{eq:Y}
\end{equation}
A map such as \eqref{eq:Y} defines a regular isotopy invariant for colored and oriented link diagrams. To see this, one can first observe that $\Y(L)$ multiplied by a factor $q^{-2j_i(j_i+1)w(L)}$ for each spin $j_i$ coloring the link $L$ is an ambiant isotopy invariant, as shown in \cite{Resh,Tur}. Indeed, it is known that two braids represent ambiant isotopic links if and only if they are related by Markov moves \cite{Birm}. The invariance under the first and second Markov moves of the renormalized value of $\Y(L)$ can be deduced respectively from \eqref{eq:enhRcommu}, which expresses the commutativity of the action of the braid group $B_n$ with $\mu^{\otimes n}$, and from the partial trace property \eqref{eq:Trleft}, which is understood diagrammatically as follows
\begin{align}
	\begin{tikzpicture}[scale=1,line width=\lw, decoration={markings, mark=at position 0.34 with {\arrow{>}}}, baseline={([yshift=-\eseq]current bounding box.center)}]
		\draw[postaction={decorate}] {[rounded corners=5pt] \ocrossr{\xd}{0}} (0,\yd) arc(0:180:\ra) -- (-2*\ra,0) arc(180:360:\ra) -- (0,0);
		\draw[rounded corners=5pt,->]  \ocrossl{0}{0};
		\node at (0.15cm,0) {\footnotesize $j$};
		\node at (\xd +0.15cm,0) {\footnotesize $j$};
	\end{tikzpicture} = q^{2j(j+1)}
	\begin{tikzpicture}[scale=1,line width=\lw, decoration={markings, mark=at position 0.34 with {\arrow{>}}}, baseline={([yshift=-\eseq+0.2cm]current bounding box.center)}]
		\draw[->] \strand{0}{0};
		\node at (0,-0.2cm) {\footnotesize $j$};
	\end{tikzpicture} \ ,
	\qquad
	\begin{tikzpicture}[scale=1,line width=\lw, decoration={markings, mark=at position 0.41 with {\arrow{>}}}, baseline={([yshift=-\eseq]current bounding box.center)}]
		\draw[postaction={decorate}] {[rounded corners=5pt] \ucrossr{\xd}{0}} (0,\yd) arc(0:180:\ra) -- (-2*\ra,0) arc(180:360:\ra) -- (0,0);
		\draw[rounded corners=5pt,->]  \ucrossl{0}{0};
		\node at (0.15cm,0) {\footnotesize $j$};
		\node at (\xd +0.15cm,0) {\footnotesize $j$};
	\end{tikzpicture} = q^{-2j(j+1)}
	\begin{tikzpicture}[scale=1,line width=\lw, decoration={markings, mark=at position 0.34 with {\arrow{>}}}, baseline={([yshift=-\eseq+0.2cm]current bounding box.center)}]
		\draw[->] \strand{0}{0};
		\node at (0,-0.2cm) {\footnotesize $j$};
	\end{tikzpicture} \ .
	\label{eq:Markov2left}
\end{align}
Since the writhe $w(L)$ is a regular isotopy invariant, it is then concluded that $\Y(L)$ as defined in \eqref{eq:Y} is also a regular isotopy invariant which satisfies 
\begin{equation}
	\Y(L^{(\pm)}) = q^{\pm 2j(j+1)}\Y(L^{(0)}), \label{eq:YRMI}
\end{equation}
where we recall that the configurations $L^{(\pm)}$ and $L^{(0)}$ are illustrated in Figure \ref{fig:onecompconf} (the components must be oriented downward in order for this figure to fit with our conventions for braids and their closure). Note that we chose to work with the regular isotopy invariant \eqref{eq:Y} instead of its renormalized ambiant isotopy version.

\begin{rem}
	The trace in the RT link invariant construction was originally \cite{Resh,Tur} defined as
	\begin{equation}
		\Tr^{(j_1,...,j_n)}(\sigma(L)(\mu^{-1})^{\otimes n}). \label{eq:Yt}
	\end{equation}
	The difference in this case is that equation \eqref{eq:YRMI} is verified with the help of the partial trace property \eqref{eq:Trright}, which can be expressed diagrammatically as follows
	\begin{align}
		\begin{tikzpicture}[scale=1,line width=\lw, decoration={markings, mark=at position 0.41 with {\arrow{>}}}, baseline={([yshift=-\eseq]current bounding box.center)}]
			\draw[postaction={decorate}] {[rounded corners=5pt] \ocrossl{0}{0}} (\xd,\yd) arc(180:0:\ra) -- (\xd+2*\ra,0) arc(360:180:\ra) -- (\xd,0);
			\draw[rounded corners=5pt,->]  \ocrossr{\xd}{0};
			\node at (-0.15cm,0) {\footnotesize $j$};
			\node at (\xd -0.15cm,0) {\footnotesize $j$};
		\end{tikzpicture} \
		= q^{2j(j+1)}
		\begin{tikzpicture}[scale=1,line width=\lw, decoration={markings, mark=at position 0.34 with {\arrow{>}}}, baseline={([yshift=-\eseq+0.2cm]current bounding box.center)}]
			\draw[->] \strand{0}{0};
			\node at (0,-0.2cm) {\footnotesize $j$};
		\end{tikzpicture} \ ,
		\qquad
		\begin{tikzpicture}[scale=1,line width=\lw, decoration={markings, mark=at position 0.34 with {\arrow{>}}}, baseline={([yshift=-\eseq]current bounding box.center)}]
			\draw[postaction={decorate}] {[rounded corners=5pt] \ucrossl{0}{0}} (\xd,\yd) arc(180:0:\ra) -- (\xd+2*\ra,0) arc(360:180:\ra) -- (\xd,0);
			\draw[rounded corners=5pt,->]  \ucrossr{\xd}{0};
			\node at (-0.15cm,0) {\footnotesize $j$};
			\node at (\xd -0.15cm,0) {\footnotesize $j$};
		\end{tikzpicture} \ 
		= q^{-2j(j+1)}
		\begin{tikzpicture}[scale=1,line width=\lw, decoration={markings, mark=at position 0.34 with {\arrow{>}}}, baseline={([yshift=-\eseq+0.2cm]current bounding box.center)}]
			\draw[->] \strand{0}{0};
			\node at (0,-0.2cm) {\footnotesize $j$};
		\end{tikzpicture} \ .
		\label{eq:Markov2right}
	\end{align}  
	Therefore, the interpretation of \eqref{eq:Yt} is simply that the braid $\sigma(L)$ is closed on the right (as is usually done by convention) instead of the left. For the quantum group $\Usl$, both quantum traces \eqref{eq:Y} and \eqref{eq:Yt} lead to the same invariant because the irreducible spin representations of $\Usl$ are isomorphic to their dual representations.   
\end{rem}

The invariant $\Y(L)$ obtained in the mathematical framework of Yang--Baxter operators and quantum groups appears to be related to the invariant $\E(L)$ obtained by means of the Chern--Simons quantum field theory (see discussions in \cite{Gua,GMM1,GMM2,GMM3,MS,RT} for instance). In the case where all the representations are spin $1/2$, it is known that both invariants lead to the Jones polynomial \cite{Tur,Wit}. Here are some arguments in support of the claim $\E(L)=\Y(L)$ in general. First, consider the unknot $U$ associated to a spin $j$. This unknot can be represented by the closure of the identity braid with one strand. The associated value of the invariant $\eqref{eq:Y}$ is  
\begin{equation}
	\Y(U;j)=\Tr^{(j)}(\mu) = \sum_{m=-j}^j \bra{j,m}q^{2H}\ket{j,m} = \sum_{m=-j}^j q^{2m} = [2j+1]_q,
\end{equation}
which is equal to the value $\E(U;j)$ as seen from the result \eqref{eq:Eunknot}. Moreover, $\Y(L)$ and $\E(L)$ have the same properties under a change of the form $L^{(\pm)} \leftrightarrow L^{(0)}$ (see equations \eqref{eq:ERMI} and \eqref{eq:YRMI}). Finally note that the eigenvalues of the representation of the braided universal $R$-matrix of $\Usl$ (see \cite{Resh} for instance) are the same as those of the half-monodromy matrix of the CS theory with gauge group $SU(2)$, which lead to the same skein relations.

\subsection{Casimir elements and Askey--Wilson relations}\label{ssec:CasimirAW}

In this section, we will use the trace \eqref{eq:Y} and its partial versions to show that the tangle diagrams \eqref{eq:Ci}--\eqref{eq:C123} also lead to the Askey--Wilson algebra in the RT construction.   

We start by considering the tangle diagram where a single strand in any spin representation of $\Usl$ is encircled by a loop in the spin-$1/2$ representation, illustrated by:
\begin{equation}
	\begin{tikzpicture}[scale=1,line width=\lw, baseline={([yshift=-\eseq]current bounding box.center)}]
		\circs{0}{\ly+0.05cm}
		\draw[->] \cstrand{0}{0};
	\end{tikzpicture} \ . \label{eq:tangleC}
\end{equation}
This tangle is obtained by taking the partial closure of the braid $\sigma_1^2 \in B_2$, as indicated below:
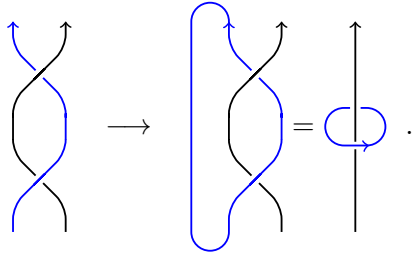
\begin{center}
	\begin{tikzpicture}[scale=1,line width=\lw,baseline={([yshift=-\eseq]current bounding box.center)}]
		\draw[\colfund,rounded corners=5pt,->] \ocrossl{0}{0} -- ++\ocrossr{0}{0};
		\draw[rounded corners=5pt,->] \ocrossr{\xd}{0} -- ++\ocrossl{0}{0};
	\end{tikzpicture} $ \quad \longrightarrow \quad $ 
	\begin{tikzpicture}[scale=1,line width=\lw, decoration={markings, mark=at position 0.46 with {\arrow{>}}}, baseline={([yshift=-\eseq]current bounding box.center)}]
		\draw[\colfund,postaction={decorate}] {[rounded corners=5pt] \ocrossl{0}{0} -- ++\ocrossr{0}{0}} (0,\yd+\yd) arc(0:180:\ra) -- (-2*\ra,0) arc(180:360:\ra) -- (0,0);
		\draw[rounded corners=5pt,->] \ocrossr{\xd}{0} -- ++\ocrossl{0}{0};
	\end{tikzpicture} $=$ 
	\begin{tikzpicture}[scale=1,line width=\lw, baseline={([yshift=-\eseq]current bounding box.center)}]
		\circs{0}{\yd-\ra}
		\draw[->] (0,0) -- ++(0,\ly+\ra+0.05cm) -- ++\cstrand{0}{0} -- ++(0,\ly+\ra+0.05cm);
	\end{tikzpicture} \ . 
\end{center}
Therefore, in the RT construction, the operator associated to the tangle \eqref{eq:tangleC} is the partial trace $\Tr_1^{(\frac{1}{2})}(\check{\cR}_{12}^2 (\mu \otimes 1))=\Tr_1^{(\frac{1}{2})}(\cR_{21}\cR_{12}(\mu \otimes 1))$. The element $\cR_{21}\cR_{12}$ commutes with the coproduct of any element of $\Usl$, as can be deduced from \eqref{eq:comult}. It is known since the inception of quantum groups \cite{Dr} that the quantum partial trace of such an element belonging to the centralizer of the diagonal action of $\Usl$ in $\Usl^{\otimes 2}$ is a central element of $\Usl$ (see also \cite{Etin,GZB,ZGB}). In the following proposition, we compute explicitly the central element associated to the tangle \eqref{eq:tangleC}.
\begin{prop}\label{prop:TrC}
	The universal $R$-matrix of $\Usl$ satisfies the following property
	\begin{equation}
		\Tr_1^{(\frac{1}{2})}(\cR_{21}\cR_{12}(\mu \otimes 1))=\Q, \label{eq:TrC}
	\end{equation}
	where $\Q$ is the Casimir element defined in \eqref{eq:Casimir}.
\end{prop}
\begin{proof}
	Since the first space on which the operator in the trace of \eqref{eq:TrC} acts is a spin $1/2$ carrier, we can write
	\begin{equation}
		\Tr_1^{(\frac{1}{2})}(\cR_{21}\cR_{12}(\mu \otimes 1))=\Tr_1(L^+_{12}L^-_{12}(M \otimes 1)).
	\end{equation}
	Using the explicit expressions \eqref{eq:Lm}, \eqref{eq:Lp} and \eqref{eq:repmu}, we obtain
	\begin{equation}
		\Tr_1^{(\frac{1}{2})}(\cR_{21}\cR_{12}(\mu \otimes 1))=
		q^{2H+1}+q^{-2H-1}+ (q-\qi)^2FE,
	\end{equation}
	which is precisely the Casimir element of $\Usl$.
\end{proof}
As a corollary, we can now obtain the operator associated to the tangle where two strands are encircled by a blue loop and which can be obtained by taking the partial closure of the braid $\sigma_1 \sigma_2^2\sigma_1 \in B_3$, as illustrated below.
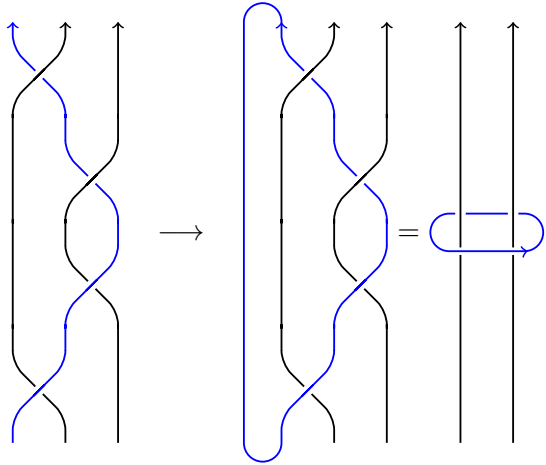
\begin{center}
	\begin{tikzpicture}[scale=1,line width=\lw,baseline={([yshift=-\eseq]current bounding box.center)}]
		\draw[\colfund,rounded corners=5pt,->] \ocrossl{0}{0} -- ++\ocrossl{0}{0} -- ++\ocrossr{0}{0} -- ++\ocrossr{0}{0};
		\draw[rounded corners=5pt,->] \ocrossr{\xd}{0} -- ++\strand{0}{0} -- ++\strand{0}{0} -- ++\ocrossl{0}{0};
		\draw[rounded corners=5pt,->] \strand{2*\xd}{0} -- ++\ocrossr{0}{0} -- ++\ocrossl{0}{0} -- ++\strand{0}{0};
	\end{tikzpicture} $ \quad \longrightarrow \quad $ 
	\begin{tikzpicture}[scale=1,line width=\lw,baseline={([yshift=-\eseq]current bounding box.center)},decoration={markings, mark=at position 0.515 with {\arrow{>}}}]
		\draw[\colfund,postaction={decorate}] {[rounded corners=5pt] \ocrossl{0}{0} -- ++\ocrossl{0}{0} -- ++\ocrossr{0}{0} -- ++\ocrossr{0}{0}} (0,\yd+\yd+\yd+\yd) arc(0:180:\ra) -- (-2*\ra,0) arc(180:360:\ra) -- (0,0);
		\draw[rounded corners=5pt,->] \ocrossr{\xd}{0} -- ++\strand{0}{0} -- ++\strand{0}{0} -- ++\ocrossl{0}{0};
		\draw[rounded corners=5pt,->] \strand{2*\xd}{0} -- ++\ocrossr{0}{0} -- ++\ocrossl{0}{0} -- ++\strand{0}{0};
	\end{tikzpicture} $=$ 
	\begin{tikzpicture}[scale=1,line width=\lw, baseline={([yshift=-\eseq]current bounding box.center)}]
		\circss{0}{\yd+\yd-\ra}{decorate}
		\draw[->] (0,0) -- ++(0,\yd+\ly+\ra+0.05cm) -- ++\cstrand{0}{0} -- ++(0,\yd+\ly+\ra+0.05cm);
		\draw[->] (\xd,0) -- ++(0,\yd+\ly+\ra+0.05cm) -- ++\cstrand{0}{0} -- ++(0,\yd+\ly+\ra+0.05cm);
	\end{tikzpicture}  
\end{center}
Similarly as before, the associated operator is $\Tr_1^{(\frac{1}{2})}(\cR_{21}\cR_{31}\cR_{13}\cR_{12} (\mu \otimes 1 \otimes 1))$. 
\begin{coro}\label{coro:TrDeltaC}
	The universal $R$-matrix of $\Usl$ satisfies the following property
	\begin{equation}
		\Tr_1^{(\frac{1}{2})}(\cR_{21}\cR_{31}\cR_{13}\cR_{12} (\mu \otimes 1 \otimes 1))=\Delta(\Q). \label{eq:TrDC}
	\end{equation}
\end{coro} 
\begin{proof}
	Take the coproduct of \eqref{eq:TrC} to get
	\begin{equation}
		\Delta\left( \Tr_1^{(\frac{1}{2})}(\cR_{21}\cR_{12}(\mu \otimes 1) (\mu \otimes 1))\right)=\Delta(\Q).
	\end{equation}
	This equation can be written as
	\begin{equation}
		\Tr_1^{(\frac{1}{2})}((\id \otimes \Delta)(\cR_{21}\cR_{12}(\mu \otimes 1)))=\Delta(\Q).
	\end{equation}
	The coproduct is a homomorphism, therefore
	\begin{equation}
		\Tr_1^{(\frac{1}{2})}((\id \otimes \Delta)(\cR_{21})(\id \otimes \Delta)(\cR_{12})(\id \otimes \Delta)(\mu \otimes 1))=\Delta(\Q).
	\end{equation}
	One can show that the second relation in \eqref{eq:idDR} implies that $(\id \otimes \Delta)(\cR_{21})=\cR_{21}\cR_{31}$. Using this and the first relation in \eqref{eq:idDR}, one gets the result \eqref{eq:TrDC}.
\end{proof}

Proposition \ref{prop:TrC} means that, in the RT link invariant construction, an open straight braid strand enclosed by a loop (that is, a closed braid strand) associated to the spin-$1/2$ representation of $\Usl$ corresponds to the Casimir element of $\Usl$. Moreover, Corollary \ref{coro:TrDeltaC} shows that adding a straight braid strand on the right of another strand inside a spin-$1/2$ loop corresponds to taking the coproduct of the Casimir element. This process of adding straight strands inside a spin-$1/2$ loop can be repeated as many times as wished, and corresponds to applying repeatedly the coproduct on the Casimir element. 

Adding a straight strand on the right of a set of strands enclosed by a loop in a diagram simply corresponds algebraically to adding a tensor factor of the identity on the right inside the associated partial traces expressions. Hence, from the results of Proposition \ref{prop:TrC} and Corollary \ref{coro:TrDeltaC}, it is seen that in the RT link invariant construction, the diagrams $\A_1, \A_{12}$ and $\A_{123}$ correspond respectively to the following intermediate Casimir elements of $\Ut$:
\begin{equation}
	\Q_1:=\Q \otimes 1 \otimes 1, \qquad
	\Q_{12} := \Delta(\Q) \otimes 1, \quad
	\Q_{123} := (\id \otimes \Delta) \circ \Delta(\Q).  
\end{equation}

By regular isotopy, one has
\begin{equation}
	\begin{tikzpicture}[scale=1,line width=\lw, baseline={([yshift=-\eseq]current bounding box.center)}]
		\draw[->] \strand{0}{0};
		\circs{\xd}{\ly+0.05cm}
		\draw[->] \cstrand{\xd}{0};
	\end{tikzpicture} \ =
	\begin{tikzpicture}[scale=1,line width=\lw, baseline={([yshift=-\eseq]current bounding box.center)}]
		\draw[rounded corners=5pt,->] \ocrossl{0}{0} -- ++\strand{0}{0} -- ++\ucrossr{0}{0};
		\circs{0}{\yd+\ly+0.05cm}
		\draw[rounded corners=5pt,->] \ocrossr{\xd}{0} -- ++\cstrand{0}{0} -- ++\ucrossl{0}{0};
	\end{tikzpicture} \ = \ 
	\begin{tikzpicture}[scale=1,line width=\lw, decoration={markings, mark=at position 0.495 with {\arrow{>}}}, baseline={([yshift=-\eseq]current bounding box.center)}]
		\draw[\colfund,postaction={decorate}] {[rounded corners=5pt] \strand{0}{0} -- ++\ocrossl{0}{0} -- ++\ocrossr{0}{0}} -- ++\strand{0}{0} (0,\yd+\yd+\yd+\yd) arc(0:180:\ra) -- (-2*\ra,0) arc(180:360:\ra) -- (0,0);
		\draw[rounded corners=5pt,->] \ocrossl{\xd}{0} -- ++\strand{0}{0} -- ++\strand{0}{0} -- ++\ucrossr{0}{0};
		\draw[rounded corners=5pt,->] \ocrossr{2*\xd}{0} -- ++\ocrossr{0}{0} -- ++\ocrossl{0}{0} -- ++\ucrossl{0}{0};
	\end{tikzpicture} \ . \label{eq:tangleC2}
\end{equation}
As a consequence of Proposition \ref{prop:TrC}, the operator associated to the tangle diagram \eqref{eq:tangleC2} is
\begin{equation}
	\check{\cR}_{12}^{-1}(\Q \otimes 1)\check{\cR}_{12} = \cR_{12}^{-1}(1 \otimes \Q)\cR_{12} = 1 \otimes \Q,  
\end{equation} 
where we used the fact that $\Q$ is central in $\Usl$. Therefore, it is seen that we can associate the diagrams $\A_2$ and $\A_3$ respectively to the following intermediate Casimir elements of $\Ut$:
\begin{equation}
	\Q_2:=1 \otimes \Q \otimes 1, \qquad \Q_3:=1 \otimes 1 \otimes \Q.  
\end{equation}
   
Finally, we can obtain with a similar procedure the operators associated to the diagrams $\A_{23}$ and $\A_{13}$. Indeed, on one hand, it is shown in \cite{CGVZ} that the element
\begin{equation}
	\Q_{13} := \check{\cR}_{23}^{-1} \Q_{12} \check{\cR}_{23} \label{eq:C13}
\end{equation}
can be interpreted as the third intermediate Casimir element of $\Ut$, and that the following equality holds
\begin{equation}
	 \Q_{23} := 1 \otimes \Delta(\Q) = \check{\cR}_{12}^{-1}\check{\cR}_{23}^{-1} \Q_{12} \check{\cR}_{23}\check{\cR}_{12}. \label{eq:C23}
\end{equation}
On the other hand, one has the following equalities of diagrams (by regular isotopy)
\begin{equation}
	\A_{13}=\sigma_2^{-1} \A_{12} \sigma_2, \quad \A_{23}=\sigma_1^{-1}\sigma_2^{-1} \A_{12} \sigma_2 \sigma_1.
\end{equation} 
Therefore, it is seen that $\A_{23}$ and $\A_{13}$ are associated to the intermediate Casimir elements $\Q_{23}$ and $\Q_{13}$.

Now it is known \cite{Zh} that the intermediate Casimir elements $\Q_I$ belong to the centralizer of the diagonal action of $\Usl$ in $\Ut$ and satisfy the relations of the special Askey--Wilson algebra $\saw(3)$ under the mapping $\C_I \mapsto \Q_I$, for all $I\in\{1,2,3,12,23,13,123\}$. Therefore, we have the following observation for the RT link invariant obtained by taking a full trace of an operator acting on the tensor product of $\Usl$-representations.

\begin{thm}\label{thm:RTAW}
	The Reshetikhin--Turaev link invariant \eqref{eq:Y} based on the universal $R$-matrix of $\Usl$ does not distinguish between the linear combinations of links given by the relations \eqref{eq:AW1}--\eqref{eq:AW4} of $\saw(3)$ under the correspondence \eqref{eq:corres} for the diagrams \eqref{eq:Ci}--\eqref{eq:C123}.   
\end{thm}  

To justify Theorem \ref{thm:RTAW}, we used the equality between the partial traces associated to the tangle diagrams $\A_I$ and the intermediate Casimir elements $\Q_I$ of $\Ut$. However, it is possible to obtain the first Askey--Wilson relation \eqref{eq:AW1} by considering the partial traces associated to the tangle diagrams which appear in \eqref{eq:AW1} and by using the properties of the $R$-matrix of $\Usl$. The relations \eqref{eq:AW2} and \eqref{eq:AW3} are then implied by conjugations of the first one. This more direct demonstration of Theorem \ref{thm:RTAW} for the three first relations of $\saw(3)$, which is albeit technical presented in Appendix \ref{app:AWtraceR}, yields as a byproduct a different proof (based on the $R$-matrix formalism) that the intermediate Casimir elements of $\Ut$ satisfy the original Askey--Wilson algebra \cite{Zh}. 

We conclude this section with some remarks.  

\begin{rem}
	In addition of $\Q_{13}$, there is another element of the centralizer of $\Usl$ in $\Ut$ associated to the recoupling of the first and third factors (studied in \cite{CGVZ}):
	\begin{equation}
		\widetilde{\Q}_{13} := \check{\cR}_{12}^{-1} \Q_{23} \check{\cR}_{12} = \check{\cR}_{23} \Q_{12} \check{\cR}_{23}^{-1}. \label{eq:C131}
	\end{equation}
	According to \eqref{eq:C131}, the element $\widetilde{\Q}_{13}$ can be associated to the following diagrams
	\begin{align}
		&\widetilde{\A}_{13} = \sigma_2 \A_{12} \sigma_2^{-1} =
		\begin{tikzpicture}[scale=1,line width=\lw,baseline={([yshift=-\eseq]current bounding box.center)}]
			\draw[->] \strand{0}{0} -- ++\cstrand{0}{0} -- ++\strand{0}{0};
			\draw[rounded corners=5pt,->] \ucrossl{\xd}{0} -- ++\strand{0}{0} ++\ocrossr{0}{0};
			\draw[rounded corners=5pt,->] \ucrossr{2*\xd}{0} -- ++\cstrand{0}{0} ++\ocrossl{0}{0};
			\circss{0}{\ly+0.05cm+\yd}{decorate}
		\end{tikzpicture} =
		\begin{tikzpicture}[scale=1,line width=\lw,baseline={([yshift=-\eseq]current bounding box.center)}]
			\draw[rounded corners=5pt,->] \ocrossl{0}{0} -- ++\cstrand{0}{0} ++\ucrossr{0}{0};
			\draw[rounded corners=5pt,->] \ocrossr{\xd}{0} -- ++\strand{0}{0} ++\ucrossl{0}{0};
			\circss{\xd}{\ly+0.05cm+\yd}{decorate}
			\draw[->] \strand{2*\xd}{0} -- ++\cstrand{0}{0} -- ++\strand{0}{0};
		\end{tikzpicture} = 
		\sigma_1^{-1} \A_{23} \sigma_1. \label{eq:C131sig}
	\end{align}
	By regular isotopy invariance, we see that
	\begin{equation}
		\widetilde{\A}_{13} =
		\begin{tikzpicture}[scale=1,line width=\lw,baseline={([yshift=-\eseq]current bounding box.center)}]
			\draw[->] \cstrand{0}{0};
			\draw[->] \cstrandb{\xd}{0};
			\draw[->] \cstrand{2*\xd}{0}; 
			\circssd{0}{\ly+0.05cm}
		\end{tikzpicture} \ .
	\end{equation}
	Note that $\Q_{13}$ is denoted by $C_{13}^{(0)}$ in \cite{CGVZ}, and $\widetilde{\Q}_{13}$ is denoted by $C_{13}^{(1)}$. Let us mention here that the equalities in \eqref{eq:C130sig} and \eqref{eq:C131sig} provide a diagrammatic interpretation of the intermediate Casimir elements $\Q_{13}$ and $\widetilde{\Q}_{13}$ of $\Ut$ as conjugations by braided universal $R$-matrices of $\Q_{12}$ or $\Q_{23}$, as studied in \cite{CGVZ}.
\end{rem}

\begin{rem}
	The results obtained in Proposition \ref{prop:TrC} and Corollary \ref{coro:TrDeltaC} are ``universal'', in the sense that there is no need to make use of a specific $\Usl$ representation (except for the blue loop since it has to be traced out). The assignment of a spin to a strand corresponds to representing the associated factor of $\Usl$. Recall that the Casimir element of $\Usl$ is represented by the identity operator times the constant     $\chi_s=q^{2s+1}+q^{-2s-1}$ in the spin-$s$ irreducible representation. Hence, we can write in the RT link invariant construction
	\begin{equation}
		\begin{tikzpicture}[scale=1,line width=\lw, baseline={([yshift=-\eseq]current bounding box.center)}]
			\draw[->] \cstrand{0}{0};
			\circs{0}{\ly+0.05cm}
			\node at (0,-0.2cm) {\footnotesize $s$};
		\end{tikzpicture} = \chi_s
		\begin{tikzpicture}[scale=1,line width=\lw, baseline={([yshift=-\eseq]current bounding box.center)}]
			\draw[->] \strand{0}{0};
			\node at (0,-0.2cm) {\footnotesize $s$};
		\end{tikzpicture} \ . \label{eq:diagCrep}
	\end{equation}   
	It can be verified that the closure of equation \eqref{eq:diagCrep} is consistent with the value of the unknot \eqref{eq:Eunknot} and the Hopf link \eqref{eq:EHopf} in the CS theory, i.e. 
	\begin{equation}
		\begin{tikzpicture}[scale=1,line width=\lw, baseline={([yshift=-\eseq]current bounding box.center)}]
			\draw[decoration={markings, mark=at position 0.3 with {\arrow{>}}}, postaction={decorate}] \cstrand{0}{0} arc(0:180:\ra) -- (-2*\ra,0) arc(180:360:\ra) -- (0,0);
			\circs{0}{\ly+0.05cm}
			\node at (0.15cm,-0.05cm) {\footnotesize $s$};
		\end{tikzpicture} 
		= [2(2s+1)]_q = \chi_s [2s+1]_q = \chi_s \ 
		\begin{tikzpicture}[scale=1,line width=\lw, baseline={([yshift=-\eseq]current bounding box.center)}]
			\draw[decoration={markings, mark=at position 0.32 with {\arrow{>}}}, postaction={decorate}] \strand{0}{0} arc(0:180:\ra) -- (-2*\ra,0) arc(180:360:\ra) -- (0,0);
			\node at (0.15cm,-0.05cm) {\footnotesize $s$};
		\end{tikzpicture}.
	\end{equation} 
	The same check can be done for several strands inside a blue loop, by using the fusion property of the Wilson lines and the knowledge of the eigenvalues of the intermediate Casimir elements of $\Ut$.
\end{rem}

\begin{rem}
	If the orientation of the blue loop in \eqref{eq:tangleC} is inverted, then one must consider the quantum partial trace of the element $\cR_{12}^{-1}\cR_{21}^{-1}$ instead in \eqref{eq:TrC}. Using the antipode $S$ defined by \eqref{eq:antipode}, the property \eqref{eq:uRi} for the inverse of the universal $R$-matrix and the defining relations \eqref{eq:Uqsl2} of $\Usl$, one can proceed as in Proposition \ref{prop:TrC} to show that the result of the quantum partial trace is again the Casimir element $\Q$ of $\Usl$. This is in agreement with the fact that the link invariant associated to the quantum algebra $\Usl$ does not depend on the orientation since the irreducible representations are isomorphic to their duals.
\end{rem}

\section{Conclusion} \label{sec:concl}

In summary, we considered the tangle diagrams where a subset of three vertical strands associated to any spin representations of $\su_2$, or $\Usl$, is encircled by a loop associated to the spin $1/2$ representation. We showed that in the Chern--Simons theory on $\bR^3$ with gauge group $SU(2)$, the Wilson loop vacuum expectation values of framed and colored links which differ in some region by a product of some of these tangle diagrams are related by the special Askey--Wilson algebra. Moreover, the same tangle diagrams were examined in the Reshetikhin--Turaev link invariant construction with quantum group $\Usl$. We have shown that the quantum partial traces of the braids representing these tangle diagrams, computed using the universal $R$-matrix of $\Usl$, are equal to the intermediate Casimir elements of $\Ut$. We found from this result that the RT link invariant also obeys the special Askey-Wilson algebra, with the generators of the algebra associated to the tangle diagrams. 

The present work opens the path for some further investigations. An obvious generalization would be to consider tangle diagrams with $n$ strands, for $n$ an integer larger than three. According to the results of this paper, such tangle diagrams are associated to intermediate Casimir elements of $\Usl^{\otimes n}$. Therefore, there should be connections between the link invariants of the CS theory and of the RT construction, and a generalized Askey--Wilson algebra $AW(n)$ \cite{DDV,DCler,PW}. Another idea would be to reproduce the same analysis for different gauge groups in the CS theory, and for their corresponding quantum groups in the RT construction. A specific example would be to consider $SU(3)$, to determine if similar tangle diagrams are also associated to natural elements of the centralizers of $U_q(\su_3)$ in its tensor products, and to look if the approach using tangles and partial traces allows to understand better the algebra formed by these elements. The study of the algebraic structure of the centralizer of $U(\su_3)$ has been initiated in \cite{CPV1,CPV2}. It could also be interesting to examine how the choice of a different manifold for the CS action affects the results of this paper. For instance, if $S^3$ is considered instead of $\bR^3$, then the coupling parameter $\kappa$ has to be an integer in order for the theory to remain invariant under gauge transformations, and hence the deformation parameter $q$ must be a root of unity. Finally, the algebra formed by the tangle diagrams considered in this paper together with the braid diagrams could be studied on its own. 
This algebra is certainly of interest and should be related for instance to the centralizers of $\Usl$ and to orthogonal polynomials \cite{CVZ1,CVZ2}.
We plan to investigate these aspects for future works.

\vspace{1em}
\noindent{\bf Acknowledgments:} The authors are grateful to Loïc Poulain d'Andecy for enlightening discussions, relevant observations and pointing out useful references. They also acknowledge helpful conversations with Julien Gaboriaud. NC thanks the CRM for its hospitality and is supported by the international research projects AAPT of the CNRS. The work of LV is supported by a Discovery Grant from the Natural Sciences and Engineering Research Council (NSERC) of Canada. MZ holds an Alexander--Graham--Bell graduate scholarship from NSERC.

\appendix

\section{Proof of Proposition \ref{prop:TrEnh}} \label{app:proofPropTrEnh}
This appendix presents a direct proof of Proposition \ref{prop:TrEnh}.

From the definition \eqref{eq:comult} of the coproduct of $\Usl$, one finds $\Delta(\mu)=\Delta^{\text{op}}(\mu)=\mu \otimes \mu$. One also finds from the property \eqref{eq:RD} that $\cR\Delta(\mu) = \Delta^{\text{op}}(\mu) \cR$. These two facts together with the action \eqref{eq:actperm} of the permutation operator and the definition \eqref{eq:cR} imply equation \eqref{eq:enhRcommu}.

For $m,n=-j,...,j$, one can write
\begin{equation}
	\bra{j,m} \Tr_1^{(j)}(\check \cR(\mu \otimes 1)) \ket{j,n} = \sum_{\ell=-j}^{j} \bra{j,m;j,\ell}\cR(\mu \otimes 1)\ket{j,\ell;j,n}.  \label{eq:Tr1Rjm}
\end{equation}
Using the explicit form \eqref{eq:uR} of the universal $R$-matrix, the orthogonality and the completeness relation \eqref{eq:orthcompl} of the basis vectors, and the fact that $\bra{j,m}H\ket{j,n} = n \delta_{mn}$ as seen from \eqref{eq:actH}, one can write the RHS of the previous equation as
\begin{align}
	\sum_{\ell=-j}^{j}\sum_{k=0}^{\infty}\frac{(q-q^{-1})^k}{[k]_q!} q^{-k(k+1)/2} &\bra{j,m}F^k\ket{j,\ell} \bra{j,\ell}E^k\ket{j,n} q^{2\ell n+(k+2)\ell - kn}. \label{eq:Rmu1}
\end{align}
For $k=0,1,...$, one can show by induction from the actions \eqref{eq:actE} and \eqref{eq:actF} the following
\begin{align}
	&\bra{j,m}E^k\ket{j,n} =
	\begin{cases}
		\frac{[j-n]_q!}{[j-n-k]_q!} & \text{if } m-n=k, \\
		0 & \text{otherwise},	
	\end{cases} \\
	&\bra{j,m}F^k\ket{j,n} =
	\begin{cases}
		\frac{[j+n]_q!}{[j+n-k]_q!} & \text{if } n-m=k, \\
		0 & \text{otherwise}.	
	\end{cases}
\end{align}
Hence the only non-zero terms in the sum \eqref{eq:Rmu1} are such that $m=n$ and $\ell=m+k$, and this sum becomes
\begin{equation}
	\delta_{mn}\sum_{k=0}^{j-m}\frac{(q-q^{-1})^k}{[k]_q!} q^{-k(k+1)/2} \frac{[j+m+k]_q! \ [j-m]_q!}{[j+m]_q! \ [j-m-k]_q!}q^{2(m+k)m+(k+2)(m+k) - km}.
\end{equation}
In terms of $q$-Pochhammer symbols (see \cite{GR} for instance), this can be rewritten as
\begin{equation}
	\delta_{mn}q^{2 m (m+1)}\sum_{k=0}^{j-m}(-1)^kq^{k(k+1+2m-2j)}\frac{(q^2;q^2)_{j-m} (q^2;q^2)_{j+m+k}}{(q^2;q^2)_k(q^2;q^2)_{j-m-k}(q^2;q^2)_{j+m}}. \label{eq:Rmu2}
\end{equation}
Using the following two identities \cite{GR}
\begin{equation}
	\frac{(q;q)_n}{(q;q)_{n-k}} = (q^{-n};q)_k(-1)^kq^{nk-\binom{k}{2}}, \qquad \frac{(a;q)_{n+k}}{(a;q)_n}=(aq^n;q)_k,
\end{equation}
one can write \eqref{eq:Rmu2} as
\begin{equation}
	\delta_{mn}q^{2 m (m+1)}\sum_{k=0}^{j-m}\frac{(q^{-2(j-m)};q^2)_k (q^{2(j+m+1)};q^2)_{k}}{(q^2;q^2)_k}q^{2k}. \label{eq:Rmu3}
\end{equation}
Finally, by applying to \eqref{eq:Rmu3} the $q$-hypergeometric function formula \cite{GR}
\begin{equation}
	{}_{2}\phi_{1}\biggl(\genfrac..{0pt}{}{q^{-n},b}{c};q,q\biggr) = \frac{(c/b;q)_n}{(c;q)_n}b^n
\end{equation}
with the substitutions $q\to q^2$, $n=j-m$, $b=q^{2(j+m+1)}$ and $c=0$, one gets 
\begin{equation}
	\bra{j,m} \Tr_1^{(j)}(\check \cR(\mu \otimes 1)) \ket{j,n} = q^{2j(j+1)}\delta_{mn}, \label{eq:Rmu4}
\end{equation}
which proves the case with positive exponents of \eqref{eq:Trleft}. A similar procedure proves the same case for \eqref{eq:Trright}. Note that if one already knows (from previous results, see \cite{Etin,GZB,ZGB} for instance) that the quantum partial trace in \eqref{eq:Tr1Rjm} is central, then the proof presented above simplifies by choosing to evaluate expression \eqref{eq:Tr1Rjm} for $m=n=j$ only. 

To obtain the case with negative exponents in \eqref{eq:Trleft}, one can first write
\begin{equation}
	\check \cR^{-1}(\mu \otimes 1) =  \cR^{-1}(1 \otimes \mu)\Pi_{12}. \label{eq:cRimu1}
\end{equation}
The inverse of the universal $R$-matrix is given by \eqref{eq:uRi}. Using the definition \eqref{eq:antipode} of the antipode $S$ and the explicit expression $\eqref{eq:uR}$ of the universal $R$-matrix, one finds that $\cR^{-1}$ is the same as $\cR$ with the replacement $q \to \qi$. It is also obviously true that the inverse of $\mu=q^{2H}$ is the same as $\mu$ with the replacement $q \to \qi$. Therefore, one finds from \eqref{eq:cRimu1}
\begin{equation}
	\left(\check \cR^{-1}(\mu \otimes 1) \right) (q)=  \left( \Pi_{12}\check \cR(1 \otimes \mu^{-1})\Pi_{12} \right)(\qi),
\end{equation}
where the dependence in $q$ has been explicitly written. As a consequence, one gets
\begin{equation}
	\Tr_1^{(j)}(\check \cR^{-1}(\mu \otimes 1)) = \Tr_2^{(j)}( \check \cR(1 \otimes \mu^{-1}))(\qi) = q^{-2j(j+1)} \id.
\end{equation}
The case with negative exponents in \eqref{eq:Trright} is obtained similarly.

\section{Askey--Wilson relations from partial traces of $R$-matrices} \label{app:AWtraceR}
In this appendix, we present an alternative proof of the fact that the intermediate Casimir elements of $\Usl$ satisfy the three Askey--Wilson relations \eqref{eq:AW1}--\eqref{eq:AW3} using quantum partial traces of $R$-matrices. Note that throughout this appendix, we will often make use of the Yang--Baxter relation in the form of equations \eqref{eq:FRT1}--\eqref{eq:FRT3} and \eqref{eq:RLL4}--\eqref{eq:RLL6} without stating it explicitly. 

\subsection{Intermediate Casimir elements in terms of partial traces}

The goal of this subsection is to write each of the intermediate Casimir elements of $\Ut$ as the quantum trace on the first space of a product of universal $R$-matrices. The first space, which is traced out in the spin-$1/2$ representation, will be labeled by $a$, and the others will be labeled by numbers $1,2,3$.  

For the elements $\Q_1$ and $\Q_{12}$, this is already done in Proposition \ref{prop:TrC} and Corollary \ref{coro:TrDeltaC} (one must simply add some tensor factors of $1$ on the right if necessary). The element $\Q_{123}$ can be obtained in a similar manner, by applying $\id \otimes \Delta$ on equation \eqref{eq:TrDC}. For reference, we write here these elements in terms of the matrices $L^{\pm}$ (which we recall are universal $R$-matrices with the first factor represented in the spin-$1/2$) and the spin-$1/2$ matrix representation $M$ of the element $\mu$ which acts only on the traced space $a$ :
\begin{align}
	&\Q_1 = \Tr_{a}(L^+_{a1}L^-_{a1} M_a), \label{eq:C1Tra}\\
	&\Q_{12} = \Tr_a(L^+_{a1} L^+_{a2} L^-_{a2} L^-_{a1}   M_a), \label{eq:C12Tra} \\
	&\Q_{123} = \Tr_a(L^+_{a1} L^+_{a2} L^+_{a3} L^-_{a3} L^-_{a2} L^-_{a1} M_a). \label{eq:C123Tra}
\end{align}

The following proposition and its corollary allow to also write the elements $\Q_2$ and $\Q_{23}$ as partial traces with the traced space $a$ on the first position.   
\begin{prop} The following equation of operators acting on two $\Usl$-representations holds:
	\begin{equation}
		\Q_2 = \Tr_{a}(L^+_{a1}L^+_{a2}L_{a2}^-(L_{a1}^+)^{-1}M_a). \label{eq:C2Tr}
	\end{equation}
\end{prop}
\begin{proof}
	Using the fact that $\Q$ is central in $\Usl$ and the result of Proposition \ref{prop:TrC}, one finds
	\begin{align}
		\Q_2
		&= \cR_{12}^{-1}\Tr_{a}(L^+_{a2}L^-_{a2}M_a)\cR_{12} \\
		&= \cR_{12}^{-1}\Tr_{a}(L^+_{a2}L^-_{a2}\cR_{12}L_{a1}^+(L_{a1}^+)^{-1}M_a) \\
		&= \cR_{12}^{-1}\Tr_{a}(L^+_{a2}L^+_{a1}\cR_{12}L_{a2}^-(L_{a1}^+)^{-1}M_a) \\
		&= \Tr_{a}(L^+_{a1}L^+_{a2}L_{a2}^-(L_{a1}^+)^{-1}M_a).
	\end{align}
\end{proof}
\begin{coro} The following equation of operators acting on three $\Usl$-representations holds:
	\begin{equation}
		\Q_{23} = \Tr_a(L^+_{a1}L^+_{a2} L^+_{a3} L^-_{a3} L_{a2}^-(L_{a1}^+)^{-1}   M_a). \label{eq:C23Tra}
	\end{equation}
\end{coro}
\begin{proof}
	Apply $(\id \otimes \Delta)$ to equation \eqref{eq:C2Tr}. Then, on the RHS of the resulting equation, use the fact that $\Delta$ is a homomorphism and the properties \eqref{eq:idDL}.
\end{proof}
The partially closed braids associated to the algebraic expressions \eqref{eq:C2Tr} and \eqref{eq:C23Tra} are illustrated in Figure \ref{fig:Q2Q23} at the end of this appendix. Note that one could proceed similarly for the element $\Q_3$, but this will actually not be useful for deriving the Askey--Wilson relations later.

Finally, we will need the following proposition to express the intermediate Casimir elements associated to the recoupling of the factors $1$ and $3$ of $\Ut$. 
\begin{prop} The following equations of operators acting on three $\Usl$-representations hold:
	\begin{align}
		\Q_{13}=&\check{\cR}_{23}^{-1}\Q_{12}\check{\cR}_{23} = \Tr_{a} (L^+_{a1}L^+_{a2}L^+_{a3}L^-_{a3}(L_{a2}^+)^{-1}L_{a1}^-M_a), \label{eq:C13Tra} \\
		\widetilde{\Q}_{13}=&\check{\cR}_{23}\Q_{12}\check{\cR}_{23}^{-1} = \Tr_{a} (L^+_{a1}(L^-_{a2})^{-1}L^+_{a3}L^-_{a3}L^-_{a2}L^-_{a1}M_a). \label{eq:C13tTra}
	\end{align}
\end{prop}
\begin{proof}
	One may use the expression \eqref{eq:C12Tra} and then proceed as indicated below: 
	\begin{align}
		\check{\cR}_{23}^{-1}\Q_{12}\check{\cR}_{23} 
		&= \check{\cR}_{23}^{-1}\Tr_{a} (L^+_{a1}L^+_{a2}L^-_{a2}L_{a1}^-M_a)\check{\cR}_{23} \\
		&= \cR_{23}^{-1}\Tr_{a} (L^+_{a1}L^+_{a3}L^-_{a3}L_{a1}^-M_a)\cR_{23} \\
		&= \cR_{23}^{-1}\Tr_{a} (L^+_{a1}L^+_{a3}L^-_{a3}\cR_{23}L^+_{a2}(L_{a2}^+)^{-1}L_{a1}^-M_a) \\
		&= \cR_{23}^{-1}\Tr_{a} (L^+_{a1}L^+_{a3}L^+_{a2}\cR_{23}L^-_{a3}(L_{a2}^+)^{-1}L_{a1}^-M_a) \\
		&= \cR_{23}^{-1}\Tr_{a} (L^+_{a1}\cR_{23}L^+_{a2}L^+_{a3}L^-_{a3}(L_{a2}^+)^{-1}L_{a1}^-M_a) \\
		&= \Tr_{a} (L^+_{a1}L^+_{a2}L^+_{a3}L^-_{a3}(L_{a2}^+)^{-1}L_{a1}^-M_a),
	\end{align}
	\begin{align}
		\check{\cR}_{23}\Q_{12}\check{\cR}_{23}^{-1}
		&= \check{\cR}_{23}\Tr_{a} (L^+_{a1}L^+_{a2}L_{a2}^-L_{a1}^-M_a)\check{\cR}_{23}^{-1} \\
		&= \cR_{32}\Tr_{a} (L^+_{a1}L^+_{a3}L_{a3}^-L_{a1}^-M_a)\cR_{32}^{-1} \\
		&= \Tr_{a} (L^+_{a1}(L^-_{a2})^{-1}L^-_{a2}\cR_{32}L^+_{a3}L_{a3}^-L_{a1}^-M_a)\cR_{32}^{-1} \\
		&= \Tr_{a} (L^+_{a1}(L^-_{a2})^{-1}L^+_{a3}\cR_{32}L^-_{a2}L_{a3}^-L_{a1}^-M_a)\cR_{32}^{-1} \\
		&= \Tr_{a} (L^+_{a1}(L^-_{a2})^{-1}L^+_{a3}L^-_{a3}L_{a2}^-\cR_{32}L_{a1}^-M_a)\cR_{32}^{-1} \\
		&=\Tr_{a} (L^+_{a1}(L^-_{a2})^{-1}L^+_{a3}L^-_{a3}L_{a2}^-L_{a1}^-M_a).
	\end{align}
\end{proof}
The partially closed braids associated to the algebraic expressions \eqref{eq:C13Tra} and \eqref{eq:C13tTra} are illustrated in Figure \ref{fig:Q13}.

\subsection{Recovering the Askey--Wilson relations} \label{app:AWRmatrix}

In this subsection, we derive the Askey--Wilson relations \eqref{eq:AW1}--\eqref{eq:AW3} with the method of partial traces of (universal) $R$-matrices.
 
Using the expressions \eqref{eq:C12Tra} and \eqref{eq:C23Tra} of the previous subsection and the fact that $\Q_{12}$ commutes with the diagonal action of $\Usl$ in $\Usl \otimes \Usl$, one can write the product $\Q_{12}\Q_{23}$ as follows:
\begin{align}
	\Q_{12}\Q_{23}
	&=\Q_{12}\Tr_{a} (L^+_{a1}L^+_{a2} L^+_{a3} L^-_{a3} L_{a2}^-(L_{a1}^+)^{-1}   M_a) \label{eq:ht1} \\
	&=\Tr_{a} (L^+_{a1}L^+_{a2} L^+_{a3}\Q_{12}  L^-_{a3} L_{a2}^-(L_{a1}^+)^{-1}   M_a) \label{eq:ht2}  \\
	&=\Tr_{a} (L^+_{a1}L^+_{a2} L^+_{a3}\Tr_{b} ( L^+_{b1}L^+_{b2}L_{b2}^-L_{b1}^-M_b) L^-_{a3} L_{a2}^-(L_{a1}^+)^{-1} M_a) \label{eq:ht3}  \\ 
	&=\Tr_{ab} (R_{ba} L^+_{a1}L^+_{a2}L^+_{b1}L^+_{b2} L^+_{a3}L^-_{a3} L_{b2}^-L_{b1}^-L_{a2}^-(L_{a1}^+)^{-1}  R_{ba}^{-1} M_aM_b) \label{eq:ht4} \\
	&=\Tr_{ab}(L^+_{b1}L^+_{b2}L^+_{a1}L^+_{a2}R_{ba}L^+_{a3}L^-_{a3}R_{ba}^{-1}L_{a2}^-(L_{a1}^+)^{-1}L_{b2}^-L_{b1}^- M_aM_b). \label{eq:ht5}
\end{align}  
The partially closed braid corresponding to \eqref{eq:ht5} is illustrated in Figure \ref{fig:Q12Q23}.

In order to recover the AW relations, the idea is to simplify the two blue crossings which appear in the braid of Figure \ref{fig:Q12Q23}. Algebraically, these two crossings correspond to the $R$-matrices acting on the spaces $a$ and $b$ in \eqref{eq:ht5}. In two spins-$1/2$ representations, one can verify that  the braided $R$-matrix $\check{R}=\Pi_{12}R$ and its inverse can be written as 
\begin{equation}
	\check{R} = q^{\frac{1}{2}} - q^{-\frac{1}{2}}P, \qquad \check{R}^{-1} = q^{-\frac{1}{2}} - q^{\frac{1}{2}}P, \label{eq:cRP}
\end{equation}
where
\begin{equation}
	P =
	\begin{pmatrix}
		0 & 0   & 0  & 0 \\
		0 & \qi & -1 & 0 \\ 
		0 & -1  & q  & 0 \\
		0 & 0   & 0  & 0
	\end{pmatrix}. \label{eq:P}
\end{equation}
It is easy to verify that the matrix $P$ satisfies
\begin{align}
	&P^2=(q+\qi)P, \\
	&\Tr_{1} (P(M\otimes \bI_2)) = \bI_2. \label{eq:trE}
\end{align}
The following propositions are also useful.
\begin{prop}\label{prop:trE}
	Denote an element of $\Usl \otimes \Usl$ with the first factor represented in the spin $1/2$ representation by $F=F^\alpha \otimes f_\alpha$ (a sum with respect to $\alpha$ is understood).
	Then the following equation holds
	\begin{equation}
		P_{12}F_{23}P_{12}=P_{12} \Tr_a(F_{a3}M_{a}). \label{eq:trEfE} 
	\end{equation} 
\end{prop}
\begin{proof}	
	Write
	\begin{equation}
		P_{12}F_{23}P_{12}=(P\otimes 1)( \bI_2 \otimes F^\alpha \otimes f_\alpha)(P\otimes 1) =P(\bI_2 \otimes F^\alpha)P \otimes f_\alpha.
	\end{equation} 
	Suppose that the matrix $F^\alpha$ is given by
	\begin{equation}
		F^\alpha = 
		\begin{pmatrix}
			a & b \\
			c & d 
		\end{pmatrix}. \label{eq:F}
	\end{equation}
	Using the explicit matrix representations \eqref{eq:P} for $P$, \eqref{eq:F} for $F^\alpha$ and \eqref{eq:repmu} for $M$, it is easy to verify that
	\begin{equation}
		P(\bI_2 \otimes F^\alpha)P = (qa + \qi d) P = \Tr(F^\alpha M)P.
	\end{equation}
	Hence
	\begin{equation}
		P_{12}F_{23}P_{12}=\Tr(F^\alpha M)P \otimes f_\alpha=P \otimes \Tr(F^\alpha M) f_\alpha=P_{12} \Tr_a(F_{a3}M_a).
	\end{equation} 
\end{proof}
\begin{prop}\label{prop:ER}
	The following equations hold
	\begin{align}
		P_{12}&=L^+_{23}L^+_{13}P_{12}(L^+_{13})^{-1}(L^+_{23})^{-1}, \label{eq:RE1} \\
		&=(L^-_{23})^{-1}(L^-_{13})^{-1}P_{12}L^-_{13}L^-_{23}, \label{eq:RE2} \\
		&=L^+_{23}L^+_{13}P_{12}L^-_{13}L^-_{23}. \label{eq:RE3}
	\end{align}
\end{prop}
\begin{proof} It is straightforward to verify these relations using the explicit matrix representations \eqref{eq:Lm}, \eqref{eq:Lp} and \eqref{eq:P}. Alternatively, \eqref{eq:RE1} and \eqref{eq:RE2} can be proven algebraically by using \eqref{eq:cRP} and the Yang--Baxter equations for $L^{\pm}$.
\end{proof}
One can use \eqref{eq:cRP} to write the partial trace \eqref{eq:ht5} as follows
\begin{align}
	&\Tr_{ab}(L^+_{b1}L^+_{b2}L^+_{a1}L^+_{a2}R_{ba}L^+_{a3}L^-_{a3}R_{ba}^{-1}L_{a2}^-(L_{a1}^+)^{-1}L_{b2}^-L_{b1}^- M_aM_b) \nonumber \\
	=&\Tr_{ab} (L^+_{b1}L^+_{b2}L^+_{a1}L^+_{a2}\check{R}_{ab}L^+_{b3}L^-_{b3}\check{R}_{ab}^{-1}L_{a2}^-(L_{a1}^+)^{-1}L_{b2}^-L_{b1}^-M_aM_b) \\
	=&\Tr_{ab} (L^+_{b1}L^+_{b2}L^+_{a1}L^+_{a2}L^+_{b3}L^-_{b3}L_{a2}^-(L_{a1}^+)^{-1}L_{b2}^-L_{b1}^-M_aM_b) \label{eq:4terms} \\
	-q &\Tr_{ab} (L^+_{b1}L^+_{b2}L^+_{a1}L^+_{a2}L^+_{b3}L^-_{b3}P_{ab}L_{a2}^-(L_{a1}^+)^{-1}L_{b2}^-L_{b1}^-M_aM_b) \nonumber \\
	-q^{-1} &\Tr_{ab} (L^+_{b1}L^+_{b2}L^+_{a1}L^+_{a2}P_{ab}L^+_{b3}L^-_{b3}L_{a2}^-(L_{a1}^+)^{-1}L_{b2}^-L_{b1}^-M_aM_b) \nonumber \\
	+ &\Tr_{ab} (L^+_{b1}L^+_{b2}L^+_{a1}L^+_{a2}P_{ab}L^+_{b3}L^-_{b3}P_{ab}L_{a2}^-(L_{a1}^+)^{-1}L_{b2}^-L_{b1}^-M_aM_b). \nonumber
\end{align}
We will now show that each of the terms in equation \eqref{eq:4terms} can be written using the intermediate Casimir elements of $\Usl$.

The trace in the first term of \eqref{eq:4terms} is simplified as follows:
\begin{align}
	&\Tr_{ab} (L^+_{b1}L^+_{b2}L^+_{a1}L^+_{a2}L^+_{b3}L^-_{b3}L_{a2}^-(L_{a1}^+)^{-1}L_{b2}^-L_{b1}^-M_aM_b) \nonumber \\
	=&\Tr_{ab} (L^+_{b1}L^+_{b2}L^+_{a1}L^+_{a2}L_{a2}^-(L_{a1}^+)^{-1}L^+_{b3}L^-_{b3}L_{b2}^-L_{b1}^-M_aM_b)\\
	=&\Tr_{b} (L^+_{b1}L^+_{b2}\Tr_{a}(L^+_{a1}L^+_{a2}L_{a2}^-(L_{a1}^+)^{-1}M_a)L^+_{b3}L^-_{b3}L_{b2}^-L_{b1}^-M_b) \\
	=&\Tr_{b} (L^+_{b1}L^+_{b2}\Q_2L^+_{b3}L^-_{b3}L_{b2}^-L_{b1}^-M_b)  \\
	=&\Q_2\Tr_{b} (L^+_{b1}L^+_{b2}L^+_{b3}L^-_{b3}L_{b2}^-L_{b1}^-M_b) \\
	=&\Q_2\Q_{123}.
\end{align}
In the previous equations, we used the expressions \eqref{eq:C123Tra} and \eqref{eq:C2Tr} and the the fact that $\Q_2$ is central. The partially closed braid corresponding to the algebraic partial trace expression for $\Q_2\Q_{123}$ is illustrated in Figure \ref{fig:Q2Q123}. 

The trace in the second term of \eqref{eq:4terms} is simplified by using \eqref{eq:RE3}, \eqref{eq:RE1} and \eqref{eq:trE} as indicated:
\begin{align}
	&\Tr_{ab} (L^+_{b1}L^+_{b2}L^+_{a1}L^+_{a2}L^+_{b3}L^-_{b3}P_{ab}L_{a2}^-L_{b2}^-(L_{a1}^+)^{-1}L_{b1}^-M_aM_b) \nonumber \\
	=&\Tr_{ab} (L^+_{b1}L^+_{b2}L^+_{a1}L^+_{b3}L^-_{b3}(L_{b2}^+)^{-1}P_{ab}(L_{a1}^+)^{-1}L_{b1}^-M_aM_b) \\
	=&\Tr_{ab} (L^+_{b1}L^+_{b2}L^+_{b3}L^-_{b3}(L_{b2}^+)^{-1}(L_{b1}^+)^{-1}P_{ab}L_{b1}^+L_{b1}^-M_aM_b) \\
	=&\Tr_{b} (L^+_{b1}L^+_{b2}L^+_{b3}L^-_{b3}(L_{b2}^+)^{-1}(L_{b1}^+)^{-1}\Tr_{a}(P_{ab}M_a)L_{b1}^+L_{b1}^-M_b) \\
	=&\Tr_{b} (L^+_{b1}L^+_{b2}L^+_{b3}L^-_{b3}(L_{b2}^+)^{-1}L_{b1}^-M_b) \\
	=&\Q_{13}.
\end{align}

The trace in the third term of \eqref{eq:4terms} can be simplified in a similar manner by using the results \eqref{eq:RE3}, \eqref{eq:RE1} and \eqref{eq:trE}: 
\begin{align}
	&\Tr_{ab} (L^+_{b1}L^+_{b2}L^+_{a1}L^+_{a2}P_{ab}L^+_{b3}L^-_{b3}L_{a2}^-(L_{a1}^+)^{-1}L_{b2}^-L_{b1}^-M_aM_b) \nonumber \\
	=&\Tr_{ab} (L^+_{b1}L^+_{a1}P_{ab}(L^-_{b2})^{-1}(L^-_{a2})^{-1}L^+_{b3}L^-_{b3}L_{a2}^-(L_{a1}^+)^{-1}L_{b2}^-L_{b1}^-M_aM_b) \\
	=&\Tr_{ab} (P_{ab}L^+_{b1}(L^-_{b2})^{-1}L^+_{b3}L^-_{b3}L_{b2}^-L_{b1}^-M_aM_b) \\
	=&\Tr_{b} (\Tr_{a}(P_{ab}M_a)L^+_{b1}(L^-_{b2})^{-1}L^+_{b3}L^-_{b3}L_{b2}^-L_{b1}^-M_b) \\
	=&\Tr_{b} (L^+_{b1}(L^-_{b2})^{-1}L^+_{b3}L^-_{b3}L_{b2}^-L_{b1}^-M_b) \\
	=&\widetilde{\Q}_{13}.
\end{align}

Finally, the trace in the fourth term of \eqref{eq:4terms} is simplified with the help of equations \eqref{eq:trEfE}, \eqref{eq:RE3}, \eqref{eq:RE1} and \eqref{eq:trE}:
\begin{align}
	&\Tr_{ab} (L^+_{b1}L^+_{b2}L^+_{a1}L^+_{a2}P_{ab}L^+_{b3}L^-_{b3}P_{ab}L_{a2}^-(L_{a1}^+)^{-1}L_{b2}^-L_{b1}^-M_aM_b) \nonumber \\
	=& \Tr_{ab} (L^+_{b1}L^+_{b2}L^+_{a1}L^+_{a2}P_{ab}\Tr_c(L^+_{c3}L^-_{c3}M_c)L_{a2}^-(L_{a1}^+)^{-1}L_{b2}^-L_{b1}^-M_aM_b) \\
	=& \Tr_{ab} (L^+_{b1}L^+_{b2}L^+_{a1}L^+_{a2}P_{ab}\Q_3L_{a2}^-(L_{a1}^+)^{-1}L_{b2}^-L_{b1}^-M_aM_b) \\
	=& \Tr_{ab} (L^+_{b1}L^+_{a1}L^+_{b2}L^+_{a2}P_{ab}L_{a2}^-L_{b2}^-(L_{a1}^+)^{-1}L_{b1}^-M_aM_b)\Q_3 \\
	=& \Tr_{ab} (L^+_{b1}L^+_{a1}P_{ab}(L_{a1}^+)^{-1}L_{b1}^-M_aM_b)\Q_3 \\
	=& \Tr_{ab} (P_{ab}L^+_{b1}L_{b1}^-M_aM_b)\Q_3 \\
	=& \Tr_{b} (\Tr_{a}(P_{ab}M_a)L^+_{b1}L_{b1}^-M_b)\Q_3 \\
	=& \Tr_{b} (L^+_{b1}L_{b1}^-M_b)\Q_3 \\
	=& \Q_1\Q_3.
\end{align}
Note that we also used expression \eqref{eq:C1Tra} and the fact that $Q_3$ is central.

Therefore, by combining the previous results, we have found
\begin{equation}
	\Q_{12}\Q_{23} = \Q_2\Q_{123} -q \Q_{13} - \qi \widetilde{\Q}_{13} + \Q_1\Q_3. \label{eq:prodQ12Q23Tr}
\end{equation}
One can proceed similarly for the product $\Q_{23}\Q_{12}$, the difference being that the blue crossings in Figure \ref{fig:Q12Q23} have to be inverted. From \eqref{eq:cRP}, we see that this corresponds to taking $q \to \qi$ in \eqref{eq:prodQ12Q23Tr}. From this results, the Askey--Wilson relation \eqref{eq:AW1} is recovered. The relations \eqref{eq:AW2} and \eqref{eq:AW3} can be obtained by conjugations of \eqref{eq:AW1} with universal $R$-matrices, as shown in \cite{CGVZ}.

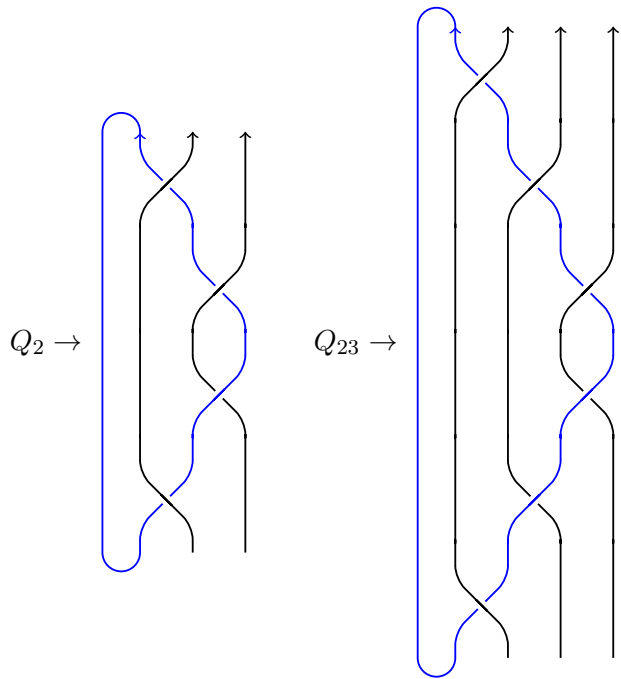
\begin{figure}[H]
	$\Q_{2} \rightarrow \ $  
	\begin{tikzpicture}[line width=\lw,baseline={([yshift=-\eseq]current bounding box.center)}]
		\draw[\colfund,decoration={markings, mark=at position 0.495 with {\arrow{>}}},postaction={decorate}]  { [rounded corners=5pt] \ucrossl{0}{0} -- ++\ocrossl{0}{0} -- ++\ocrossr{0}{0} -- ++\ocrossr{0}{0}} (0,\yd+\yd+\yd+\yd) arc(0:180:\ra) -- (-2*\ra,0) arc(180:360:\ra) -- (0,0);
		\draw[->,rounded corners=5pt]  \ucrossr{\xd}{0} -- ++\strand{0}{0} -- ++\strand{0}{0} -- ++\ocrossl{0}{0};
		\draw[->,rounded corners=5pt]  \strand{2*\xd}{0} -- ++\ocrossr{0}{0} -- ++\ocrossl{0}{0} -- ++\strand{0}{0};
	\end{tikzpicture}  $ \qquad \Q_{23} \rightarrow \ $ 
	\begin{tikzpicture}[scale=1,line width=\lw,baseline={([yshift=-\eseq]current bounding box.center)}]
		\draw[\colfund,decoration={markings, mark=at position 0.523 with {\arrow{>}}},postaction={decorate}]  { [rounded corners=5pt] \ucrossl{0}{0} -- ++\ocrossl{0}{0} -- ++\ocrossl{0}{0} -- ++\ocrossr{0}{0} -- ++\ocrossr{0}{0} -- ++\ocrossr{0}{0}} (0,\yd+\yd+\yd+\yd+\yd+\yd) arc(0:180:\ra) -- (-2*\ra,0) arc(180:360:\ra) -- (0,0);
		\draw[->,rounded corners=5pt]  \ucrossr{\xd}{0} -- ++\strand{0}{0} -- ++\strand{0}{0} -- ++\strand{0}{0} -- ++\strand{0}{0} -- ++\ocrossl{0}{0};
		\draw[->,rounded corners=5pt]  \strand{2*\xd}{0} -- ++\ocrossr{0}{0} -- ++\strand{0}{0} -- ++\strand{0}{0} -- ++\ocrossl{0}{0} -- ++\strand{0}{0};
		\draw[->,rounded corners=5pt]  \strand{3*\xd}{0} -- ++\strand{0}{0} -- ++\ocrossr{0}{0} -- ++\ocrossl{0}{0} -- ++\strand{0}{0} -- ++\strand{0}{0};
	\end{tikzpicture}
	\centering
	\caption{Partially closed braids associated to the intermediate Casimir elements $\Q_{2}$ and $\Q_{23}$.}
	\label{fig:Q2Q23}
\end{figure}

\begin{figure}[H]
	$\Q_{13} \rightarrow \ $  
	\begin{tikzpicture}[line width=\lw,baseline={([yshift=-\eseq]current bounding box.center)}]
		\draw[\colfund,decoration={markings, mark=at position 0.523 with {\arrow{>}}},postaction={decorate}]  { [rounded corners=5pt] \ocrossl{0}{0} -- ++\ucrossl{0}{0} -- ++\ocrossl{0}{0} -- ++\ocrossr{0}{0} -- ++\ocrossr{0}{0} -- ++\ocrossr{0}{0}} (0,\yd+\yd+\yd+\yd+\yd+\yd) arc(0:180:\ra) -- (-2*\ra,0) arc(180:360:\ra) -- (0,0);
		\draw[->,rounded corners=5pt]  \ocrossr{\xd}{0} -- ++\strand{0}{0} -- ++\strand{0}{0} -- ++\strand{0}{0} -- ++\strand{0}{0} -- ++\ocrossl{0}{0};
		\draw[->,rounded corners=5pt]  \strand{2*\xd}{0} -- ++\ucrossr{0}{0} -- ++\strand{0}{0} -- ++\strand{0}{0} -- ++\ocrossl{0}{0} -- ++\strand{0}{0};
		\draw[->,rounded corners=5pt]  \strand{3*\xd}{0} -- ++\strand{0}{0} -- ++\ocrossr{0}{0} -- ++\ocrossl{0}{0} -- ++\strand{0}{0} -- ++\strand{0}{0};
	\end{tikzpicture}  $ \qquad \widetilde{\Q}_{13} \rightarrow \ $ 
	\begin{tikzpicture}[scale=1,line width=\lw,baseline={([yshift=-\eseq]current bounding box.center)}]
		\draw[\colfund,decoration={markings, mark=at position 0.547 with {\arrow{>}}},postaction={decorate}]  { [rounded corners=5pt] \ocrossl{0}{0} -- ++\ocrossl{0}{0} -- ++\ocrossl{0}{0} -- ++\ocrossr{0}{0} -- ++\ucrossr{0}{0} -- ++\ocrossr{0}{0}} (0,\yd+\yd+\yd+\yd+\yd+\yd) arc(0:180:\ra) -- (-2*\ra,0) arc(180:360:\ra) -- (0,0);
		\draw[->,rounded corners=5pt]  \ocrossr{\xd}{0} -- ++\strand{0}{0} -- ++\strand{0}{0} -- ++\strand{0}{0} -- ++\strand{0}{0} -- ++\ocrossl{0}{0};
		\draw[->,rounded corners=5pt]  \strand{2*\xd}{0} -- ++\ocrossr{0}{0} -- ++\strand{0}{0} -- ++\strand{0}{0} -- ++\ucrossl{0}{0} -- ++\strand{0}{0};
		\draw[->,rounded corners=5pt]  \strand{3*\xd}{0} -- ++\strand{0}{0} -- ++\ocrossr{0}{0} -- ++\ocrossl{0}{0} -- ++\strand{0}{0} -- ++\strand{0}{0};
	\end{tikzpicture}
	\centering
	\caption{Partially closed braids associated to the intermediate Casimir elements $\Q_{13}$ and $\widetilde{\Q}_{13}$.}
	\label{fig:Q13}
\end{figure}

\begin{figure}[H]
	$\Q_{12}\Q_{23} \rightarrow \ $  
	\begin{tikzpicture}[scale=1,line width=\lw,baseline={([yshift=-\eseq]current bounding box.center)}]
		\draw[\colfund,decoration={markings, mark=at position 0.53 with {\arrow{>}}},postaction={decorate}]  { [rounded corners=5pt] \strand{0}{0} -- ++\ucrossl{0}{0} -- ++\ocrossl{0}{0} -- ++\ucrossl{0}{0} -- ++\ocrossl{0}{0} -- ++\ocrossr{0}{0} -- ++\ocrossr{0}{0} -- ++\ocrossr{0}{0} -- ++\ocrossr{0}{0}} -- ++\strand{0}{0} arc(0:180:\ra) -- (-2*\ra,0) arc(180:360:\ra) -- (0,0);
		\draw[\colfund,decoration={markings, mark=at position 0.4775 with {\arrow{>}}},postaction={decorate}] { [rounded corners=5pt] \ocrossl{\xd}{0} -- ++\ocrossl{0}{0} -- ++\strand{0}{0} -- ++\ucrossr{0}{0} -- ++\strand{0}{0} -- ++\strand{0}{0} -- ++\ocrossl{0}{0} -- ++\strand{0}{0} -- ++\ocrossr{0}{0} -- ++\ocrossr{0}{0}} arc(0:180:\ra+\xd) -- (-2*\ra-\xd,0) arc(180:360:\ra+\xd) -- (\xd,0) ; 
		\draw[rounded corners=5pt,->] \ocrossr{2*\xd}{0} -- ++\ucrossr{0}{0} -- ++\strand{0}{0} -- ++\strand{0}{0} -- ++\strand{0}{0} -- ++\strand{0}{0} -- ++\strand{0}{0} -- ++\strand{0}{0} -- ++\ocrossl{0}{0} -- ++\ocrossl{0}{0};
		\draw[rounded corners=5pt,->] \strand{3*\xd}{0} -- ++\ocrossr{0}{0} -- ++\ocrossr{0}{0} -- ++\strand{0}{0} -- ++\strand{0}{0} -- ++\strand{0}{0} -- ++\strand{0}{0} -- ++\ocrossl{0}{0} -- ++\ocrossl{0}{0} -- ++\strand{0}{0};
		\draw[rounded corners=5pt,->] \strand{4*\xd}{0} -- ++\strand{0}{0} -- ++\strand{0}{0} -- ++\strand{0}{0} -- ++\ocrossr{0}{0} -- ++\ocrossl{0}{0} -- ++\strand{0}{0} -- ++\strand{0}{0} -- ++\strand{0}{0} -- ++\strand{0}{0};
	\end{tikzpicture}
	\centering
	\caption{Partially closed braid associated to the product $\Q_{12}\Q_{23}$.}
	\label{fig:Q12Q23}
\end{figure}
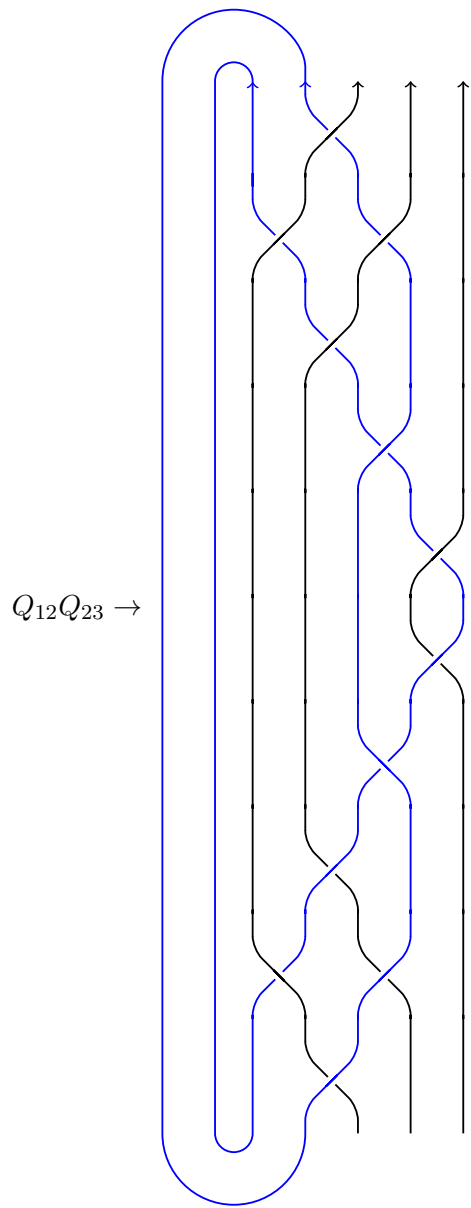  

\begin{figure}[H]
	$\Q_{2}\Q_{123} \rightarrow \ $  
	\begin{tikzpicture}[scale=1,line width=\lw,baseline={([yshift=-\eseq]current bounding box.center)}]
		\draw[\colfund,decoration={markings, mark=at position 0.51 with {\arrow{>}}},postaction={decorate}]  { [rounded corners=5pt] \strand{0}{0} -- ++\ucrossl{0}{0} -- ++\ocrossl{0}{0} -- ++\ocrossr{0}{0} -- ++\ocrossr{0}{0}} -- ++\strand{0}{0} arc(0:180:\ra) -- (-2*\ra,0) arc(180:360:\ra) -- (0,0);
		\draw[\colfund,decoration={markings, mark=at position 0.445 with {\arrow{>}}},postaction={decorate}] { [rounded corners=5pt] \ocrossl{\xd}{0} -- ++\ocrossl{0}{0}  -- ++\ocrossl{0}{0} -- ++\ocrossr{0}{0} -- ++\ocrossr{0}{0} -- ++\ocrossr{0}{0}} arc(0:180:\ra+\xd) -- (-2*\ra-\xd,0) arc(180:360:\ra+\xd) -- (\xd,0) ; 
		\draw[rounded corners=5pt,->] \ocrossr{2*\xd}{0} -- ++\ucrossr{0}{0} -- ++\strand{0}{0} -- ++\strand{0}{0} -- ++\ocrossl{0}{0} -- ++\ocrossl{0}{0};
		\draw[rounded corners=5pt,->] \strand{3*\xd}{0} -- ++\ocrossr{0}{0} -- ++\ocrossr{0}{0} -- ++\ocrossl{0}{0} -- ++\ocrossl{0}{0} -- ++\strand{0}{0};
		\draw[rounded corners=5pt,->] \strand{4*\xd}{0} -- ++\strand{0}{0} -- ++\ocrossr{0}{0} -- ++\ocrossl{0}{0} -- ++\strand{0}{0} -- ++\strand{0}{0};
	\end{tikzpicture}
	\centering
	\caption{Partially closed braid associated to the product $\Q_{2}\Q_{123}$.}
	\label{fig:Q2Q123}
\end{figure}
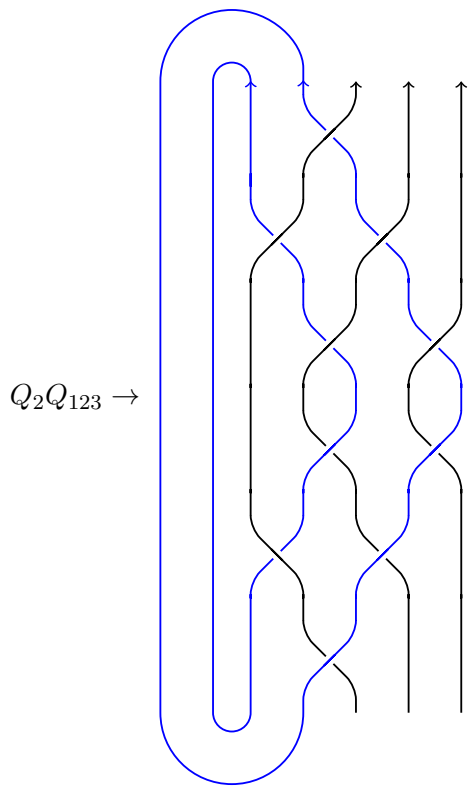


\begin{thebibliography}{99}
	
	\bibitem{A} E. Artin,
	\textsl{Theory of braids}, 
	Ann. of Math. \textbf{48} (1947) 101--126.
	
	\bibitem{Ast} M. Astorino, 
	\textsl{Kauffman knot invariant from $SO(N)$ or $Sp(N)$ Chern--Simons theory and the Potts model},
	Int. J. Mod. Phys. A \textbf{5} (1990) 1165--1195.
	
	\bibitem{Birm} J. S. Birman,
	\textsl{Braids, links and mapping class groups}, 
	Ann. Math. Stud. \textbf{82} (1974).
	
	\bibitem{BW} J. S. Birman and H. Wenzl,
	\textsl{Braids, link polynomials and a new algebra}, 
	Trans. Amer. Math. Soc. \textbf{313} (1989) 249--273.
	
	\bibitem{BP} D. Bullock and J. H. Przytycki,
	\textsl{Multiplicative structure of Kauffman bracket skein module quantizations}, 
	Proc. Amer. Math. Soc. \textbf{128} (1999) 923--931 and \href{https://arxiv.org/abs/math/9902117}{\texttt{arXiv:math/9902117}}.
	
	\bibitem{Cooke} J. Cooke,
	\textsl{Kauffman skein algebras and quantum Teichmüller spaces via factorisation homology}, 
	J. Knot Theory Ramifications \textbf{29} (2020) 2050089 and \href{https://arxiv.org/abs/1811.09293v3}{\texttt{arXiv:1811.09293v3}}.
	
	\bibitem{CGMM} P. Cotta-Ramusino, E. Guadagnini, M. Martellini and M. Mintchev, 
	\textsl{Quantum field theory and link invariants},
	Nucl. Phys. B \textbf{330} (1990) 557--574.	
	
	\bibitem{CFGPRV} N. Cramp\'e, L. Frappat, J. Gaboriaud, L. Poulain d'Andecy, E. Ragoucy and L. Vinet,
	\textsl{The Askey--Wilson algebra and its avatars}, 
	J. Phys. A: Math. Theor. \textbf{54} (2021) 063001 and \href{https://arxiv.org/abs/2009.14815}{\texttt{arXiv:2009.14815}}.
	
	\bibitem{CGVZ} N. Cramp\'e, J. Gaboriaud, L. Vinet and M. Zaimi,
	\textsl{Revisiting the Askey--Wilson algebra with the universal $R$-matrix of $U_q(sl(2))$}, 
	J. Phys. A: Math. Theor. \textbf{53} (2020) 05LT01 and \href{https://arxiv.org/abs/1908.04806}{\texttt{arXiv:1908.04806v2}}.
	
	\bibitem{CPV1} N. Cramp\'e, L. Poulain d'Andecy and L. Vinet,
	\textsl{The missing label of $\mathfrak{su}_3$ and its symmetry}, 
	\href{https://arxiv.org/abs/2110.03521}{\texttt{arXiv:2110.03521}} (2021).
	
	\bibitem{CPV2} N. Cramp\'e, L. Poulain d'Andecy and L. Vinet,
	\textsl{A Calabi-Yau algebra with $E_6$ symmetry and the Clebsch-Gordan series of $sl(3)$},
	J. Lie Theory \textbf{31} (2021) 1085--1112 and 
	\href{https://arxiv.org/abs/2005.13444}{\texttt{arXiv:2005.13444}}.
	
	\bibitem{CVZ1} N. Cramp\'e, L. Vinet and M. Zaimi, 
	\textsl{Braid group and $q$-Racah polynomials}, Proc. Amer. Math. Soc. \textbf{150} (2022) 951--966 and \href{https://arxiv.org/abs/2106.02416}{\texttt{arXiv:2106.02416v2}}.
	
	\bibitem{CVZ2} N. Cramp\'e, L. Vinet and M. Zaimi, 
	\textsl{Temperley--Lieb, Birman--Murakami--Wenzl and Askey--Wilson algebras and other centralizers of $U_q(\mathfrak{sl}_2)$}, Ann. Henri Poincar\'e \textbf{22} (2021) 3499--3528 and \href{https://arxiv.org/abs/2008.04905}{\texttt{arXiv:2008.04905}}.

	\bibitem{DDV} H. De Bie, H. De Clercq and W. van de Vijver, \textsl{The higher rank $q$-deformed Bannai--Ito and Askey--Wilson algebra}, 
	Commun. Math. Phys. \textbf{374} (2020) 277--316 and 
	\href{https://arxiv.org/abs/1805.06642}{\texttt{arXiv:1805.06642}}.
	
	\bibitem{DCler} H. De Clercq, \textsl{Higher rank relations for the Askey--Wilson and $q$-Bannai--Ito algebra}, 
	SIGMA \textbf{15} (2019) 099 and 
	\href{https://arxiv.org/abs/1908.11654}{\texttt{arXiv:1908.11654}}.
	
	\bibitem{Dr} V.G. Drinfeld,
	\textsl{Quantum groups}, 
	in: Proc. ICM (Berkeley,1986), vol 1 (Academic Press, New York, 1987) 798--820.
	
	\bibitem{Etin} P. I. Etingof,
	\textsl{Central elements for quantum affine algebras and affine Macdonald's operators}, 
	Math. Res. Lett. \textbf{2} (1995) 611--628 and \href{https://arxiv.org/abs/q-alg/9412007}{\texttt{arXiv:q-alg/9412007}}.
	
	\bibitem{FRT} L.D. Faddeev, N.Yu. Reshetikhin and L.A. Takhtajan,
	\textsl{Quantization  of  Lie  groups  and  Lie algebras},
	Leningrad Math. J. \textbf{1} (1990) 193.
	
	\bibitem{HOMFLY} P. Freyd, D. Yetter, J. Hoste, W.B.R. Lickorish, K. Millett and A. Ocneanu,
	\textsl{A new polynomial invariant of knots and links},
	Bull. Amer. Math. Soc. \textbf{12} (1985) 239--246.
	
	\bibitem{GR} G. Gasper and M. Rahman,
	\textsl{Basic hypergeometric series}, 
	Encyclopedia of Mathematics and its Applications, Cambridge University Press, 2nd edition (2004).
	
	\bibitem{GZB} M. D. Gould, R. B. Zhang and A. J. Bracken,
	\textsl{Generalized Gel’fand invariants and characteristic identities for quantum groups},
	J. Math. Phys. \textbf{32} (1991) 2298.
	
	\bibitem{GZ} Ya. A. Granovskii and A. S. Zhedanov,
	\textsl{Hidden symmetry of the Racah and Clebsch-Gordan problems for the quantum algebra $sl_q(2)$}, 
	Journal of Group Theory in Physics \textbf{1} (1993) 161--171 and \href{https://arxiv.org/abs/hep-th/9304138}{\texttt{arXiv:hep-th/9304138}}.
	
	\bibitem{Gua} E. Guadagnini, 
	\textsl{The link invariants of the Chern-Simons field theory: New developments in topological quantum field theory}, 
	De Gruyter, Berlin, New York, 1993.
	
	\bibitem{GMM4} E. Guadagnini, M. Martellini and M. Mintchev, 
	\textsl{Wilson lines in Chern-Simons theory and link invariants},
	Nucl. Phys. B \textbf{330} (1990) 575--607.
	
	\bibitem{GMM5} E. Guadagnini, M. Martellini and M. Mintchev, 
	\textsl{Link invariants from Chern-Simons theory},
	Nucl. Phys. B (Proc. Suppl.) \textbf{18} (1990) 121--134.
	
	\bibitem{GMM1} E. Guadagnini, M. Martellini and M. Mintchev, 
	\textsl{Chern--Simons holonomies and the appearance of quantum groups},
	Phys. Lett. B \textbf{235} (1990) 275--281.
	
	\bibitem{GMM2} E. Guadagnini, M. Martellini and M. Mintchev, 
	\textsl{Braids and quantum group symmetry in Chern--Simons theory},
	Nucl. Phys. B \textbf{336} (1990) 581--609.
	
	\bibitem{GMM3} E. Guadagnini, M. Martellini and M. Mintchev, 
	\textsl{Chern--Simons field theory and quantum groups},
	in: Quantum Groups, Lecture Notes in Physics, vol 370, 307--317, Springer, Berlin, Heidelberg, 1990.
	
	\bibitem{Hik} K. Hikami,
	\textsl{DAHA and skein algebra of surface: double-torus knots},
	Lett. Math. Phys. \textbf{109} (2019) 2305-–2358 and 
	\href{https://arxiv.org/abs/1901.02743}{\texttt{arXiv:1901.02743}}
	
	\bibitem{Hor} J. H. Horne, 
	\textsl{Skein relations and Wilson loops in Chern--Simons gauge theory},
	Nucl. Phys. B \textbf{334} (1990) 669--694.
	
	\bibitem{H} H.-W. Huang,
	\textsl{An embedding of the universal Askey--Wilson algebra into $U_q(\mathfrak{sl}_2)\otimes U_q(\mathfrak{sl}_2)\otimes U_q(\mathfrak{sl}_2)$}, 
	Nucl. Phys. B \textbf{922} (2017) 401--434 and 
	\href{https://arxiv.org/abs/1611.02130}{\texttt{arXiv:1611.02130}}.
	
	\bibitem{Jo85} V. F. R. Jones,
	\textsl{A polynomial invariant for knots via von Neumann algebras},
	Bull. Amer. Math. Soc. \textbf{12} (1985) 103--111.
	
	\bibitem{Jo87} V. F. R. Jones,
	\textsl{Hecke algebra representations of braid groups and link polynomials},
	Ann. Math. \textbf{126} (1987) 335--388.
	
	\bibitem{Kau90} L. Kauffman,
	\textsl{An invariant of regular isotopy},
	Trans. Amer. Math. Soc. \textbf{318} (1990) 417--471.

		
	\bibitem{Kau87} L. Kauffman,
	\textsl{State models and the Jones polynomial},
	Topology \textbf{26} (1987) 395--307.
	
	
	\bibitem{KCP} T. W. Kim, B. H. Cho and S. U. Park, 
	\textsl{Chern--Simons theories on $SO(N)$ and $Sp(2N)$ and link polynomials},
	Phys. Rev. D \textbf{42} (1990) 4135--4138.
	
	\bibitem{KR} A. N. Kirillov and N. Yu. Reshetikhin,
	\textsl{Representations of the algebra $U_q(sl(2))$, $q$-orthogonal polynomials and invariants of links}, 
	in: Infinite dimensional Lie algebras and groups (World Sci. Publishing, NJ, 1989) 285--339.
	
	\bibitem{Koek} R. Koekoek, P.A. Lesky and R.F. Swarttouw,
	\textsl{Hypergeometric orthogonal polynomials and their $q$-analogues}, 
	Springer, 1-st edition (2010).
	
	\bibitem{MS} A. Morozov and A. Smirnov,
	\textsl{Chern--Simons theory in the temporal gauge and knot invariants through the universal quantum R-matrix},
	Nucl. Phys. B \textbf{835} (2010) 284--313.
	
	\bibitem{Mur} J. Murakami,
	\textsl{The Kauffman polynomial of links and representation theory}, 
	Osaka J. Math. \textbf{24} (1987), 745--758.
	
	\bibitem{PW} S. Post and A. Walter, 
	\textsl{A higher rank extension of the Askey--Wilson algebra},
	\href{https://arxiv.org/abs/1705.01860}{\texttt{arXiv:1705.01860}} (2017).
	
	\bibitem{PT} J. H. Przytycki and P. Traczyk,
	\textsl{Invariants of links of Conway type}, 
	Kobe J. Math. \textbf{4} (1987) 115--139 and 
	\href{https://arxiv.org/abs/1610.06679}{\texttt{arXiv:1610.06679}}.
	
	\bibitem{Resh} N. Yu. Reshetikhin,
	\textsl{Quantized universal enveloping algebras, the Yang--Baxter equations and invariants of links, I and II},
	LOMI preprints E-4-87 and E-17-87, Leningrad, 1987.
	
	\bibitem{RT} N. Reshetikhin and V. G. Turaev,
	\textsl{Invariants of 3-manifolds via link polynomials and quantum groups},
	Invent. Math. \textbf{103} (1991) 547-–597.
	
	\bibitem{TL} N. Temperley and E. Lieb, 
	\textsl{Relations between the 'Percolation' and 'Colouring' Problem and other Graph-Theoretical Problems Associated with Regular Planar Lattices: Some Exact Results for the 'Percolation' Problem},
	Proc. Royal Soc. A \textbf{322} (1971) 251--280.
	
	\bibitem{Tur} V.G. Turaev,
	\textsl{The Yang--Baxter equation and invariants of links},
	Invent. Math. \textbf{92} (1988) 527-–553.
	
	\bibitem{Wit} E. Witten,
	\textsl{Quantum field theory and the Jones polynomial},
	Comm. Math. Phys. \textbf{121} (1989) 351--399.
	
	\bibitem{WY} Y.-S. Wu and K. Yamagishi, 
	\textsl{Chern--Simons theory and Kauffman polynomials},
	Int. J. Mod. Phys. A \textbf{5} (1990) 1165--1195.
	
	\bibitem{ZGB} R. B. Zhang, M. D. Gould and A. J. Bracken,
	\textsl{Quantum group invariants and link polynomials},
	Commun. Math. Phys. \textbf{137} (1991) 13--27.
	
	\bibitem{Zh}A.S. Zhedanov,
	\textsl{Hidden symmetry of the Askey--Wilson polynomials}, 
	Theor. Math. Phys. \textbf{89} (1991) 1146--1157.
	
	
\end{thebibliography}
\end{document}